\documentclass[ejs]{imsart}

\usepackage[utf8]{inputenc}
\usepackage{amsmath}
\usepackage{algorithm2e}
\usepackage{geometry}
\usepackage{mathtools}
\usepackage{amssymb}
\usepackage{amsthm}
\usepackage{multirow}
\usepackage{bm}
\usepackage{natbib}
\usepackage{color}
\RequirePackage[colorlinks,citecolor=blue,urlcolor=blue]{hyperref}

\newcommand{\floor}[1]{\left\lfloor #1 \right\rfloor}
\DeclareMathOperator{\tr}{tr}
\theoremstyle{plain}
\newtheorem{theorem}{Theorem}[section]

\newtheorem{proposition}{Proposition}[section]
\newtheorem{assumption}{Assumption}[section]

\newtheorem{lemma}{Lemma}[section]

\DeclareMathOperator*{\argmax}{arg\,max}

\newcommand{\change}[1]{{\color{black}{#1}}}
\newcommand{\csv}[1]{{\color{cyan}{#1}}}

\newcommand{\Dkonv}{\stackrel{d}{\longrightarrow}}

\newcommand{\weak}{\rightsquigarrow}

\newcommand{\R}{\mathbb{R}}
\newcommand{\eps}{\varepsilon}

\begin{document}

\begin{frontmatter}

\title{Change-point Inference for High-dimensional Heteroscedastic Data}

\runtitle{Change-point Inference for Heteroscedastic Data}

\begin{aug}
\author{\fnms{Teng} \snm{Wu}\ead[label=e1]{tengwu2@illinois.edu}}

\address{Microsoft Corporation\\
Redmond\\
\printead{e1}}

\author{\fnms{Stanislav} \snm{Volgushev}
\ead[label=e2]{stanislav.volgushev@utoronto.ca}
}

\address{Department of Statistics\\
University of Toronto\\
\printead{e2}}

\author{\fnms{Xiaofeng} \snm{Shao}\thanksref{t3}
\ead[label=e3]{xshao@illinois.edu}
}

\address{Department of Statistics\\
University of Illinois a Urbana-Champaign\\
\printead{e3}\\
}

\thankstext{t3}{Stanislav Volgushev was supported by a grant from NSERC of Canada;
Xiaofeng Shao was supported by grants NSF-DMS2014018  and NSF-DMS2210002. We would like to thank two referees for constructive comments, which led to substantial improvements.}
\runauthor{T. Wu et al.}

\end{aug}

\begin{abstract}
    We propose a bootstrap-based test to detect a mean shift in a sequence of high-dimensional observations with unknown time-varying heteroscedasticity. 
    The proposed test builds on the U-statistic based approach in \cite{wang2019inference}, targets a dense alternative, and adopts a wild bootstrap procedure to generate critical values. The bootstrap-based test is free of tuning parameters and is capable of accommodating unconditional time varying heteroscedasticity in the high-dimensional observations, as demonstrated in our theory and simulations. Theoretically, we justify the bootstrap consistency by using the recently proposed unconditional approach in \cite{bucher2019note}. 
    Extensions to testing for multiple change-points and estimation using wild binary segmentation are also presented. 
    Numerical simulations demonstrate the robustness of the proposed testing and estimation procedures with respect to different kinds of time-varying heteroscedasticity. 
    
\end{abstract}


\tableofcontents
\end{frontmatter}

\section{Introduction}

Owing to the advances in science and technology, high-dimensional data has been increasingly important in many areas, such as genomics, neuroscience and finance among others. In the analysis of high-dimensional datasets, often some kind of homogeneity assumption such as iid (independent and identically distributed) is made, but in reality the data may exhibit certain breaks in its stochastic property, especially when the data is ordered by time (e.g., stock return data) or one-dimensional locations (e.g., gene expression levels indexed by genomic loci). This has motivated a growing literature of change-point testing and estimation for the mean shift in high-dimensional data. 
See  \cite{horvath2012change,cho2015multiple,jirak2015uniform,wang2018high,wang2019inference, enikeeva2019high,  yu2021, zhangwangshao2021} for some recent work.

A common feature of all above-mentioned papers is that they assume the second order properties (i.e, covariance matrix for independent high-dimensional data) is time invariant, while the mean may undergo changes at unknown times. This is a strong assumption and may be violated for many high-dimensional datasets. See Section~\ref{sec:data} for significant evidence of time varying heteroscedasticity for a genomic dataset that has been analyzed by several researchers [\cite{wang2018high,wang2019inference,zhangwangshao2021}]. 
When heteroscedasticity is present, the existing change-point detection methods developed under the homoscedastic assumption may fail or their validity remains unknown. 
For low dimensional time series,  novel change-point detection methods
have been developed by \cite{zhou2013heteroscedasticity} and \cite{gorecki2018change} to detect mean changes while allowing for second or higher order non-stationarity, but an extension of their methods to high-dimensional setting is very nontrivial. In summary, there is a lack of methodology to detect mean changes for high-dimensional heteroscedastic data.

In this article, we develop a novel test and estimation procedure that can detect 
change-points in the mean when unconditional heteroscedasticity is present in the sequence of high-dimensional observations. To facilitate our methodological development, we assume the following mathematical framework: the $p$-dimensional observation at the $i$th time or location is 
\begin{eqnarray}
\label{model1}
X_i = \mu_i + H(i/n)Z_i, ~i=1,\cdots,n,
\end{eqnarray}
where $Z_i$'s are i.i.d. $p$-dimensional random vectors with mean 0 and covariance matrix $\Sigma$, and $H(i/n)$ is a $p \times p$ diagonal matrix that models the unconditional time/location dependent heteroscedasticity.
We are interested in testing 
\[H_0: \mu_1 = \mu_2 \ldots = \mu_n \text{ vs }\]
\change{
$$
H_{1}: \exists s \in \mathbb{N} \text { and } 1<k_{1}<\cdots<k_{s}< n \text { such that }
$$
$$
\mu_{1}=\cdots=\mu_{k_{1}} \neq \mu_{k_{1}+1}=\ldots \cdots \ldots=\mu_{k_{s}} \neq \mu_{k_{s}+1}=\cdots=\mu_{n} .
$$
Under $H_1$, $k_1 ,\ldots, k_s$ are unknown change points. }
The estimation of the number $s$ and location of change-points $(k_1,\cdots,k_s)$ is also addressed in the present paper.  
Note that when $p=1$, our model is similar to that in \cite{gorecki2018change}, except that the latter paper allowed serial dependence in $\{Z_i\}_{i=1}^{n}$. We do not pursue the more general, heteroscedastic and temporally dependent case, as there are methodological challenges to handle temporal dependence in the high-dimensional setting; see Section~\ref{sec:conclusion} for more discussions. Nevertheless, the temporal independence assumption is commonly made in the literature of change-point detection of genomic data; see \cite{jeng10} and \cite{zhang12}. 



In this paper,  we propose to build on the U-statistic based approach in \cite{wang2019inference}, who extended the two sample U-statistic used in \cite{chen2010two} from high-dimensional two sample testing to change-point testing. In \cite{wang2019inference}, the sequence of observations is assumed to be homoskedastic 
subject to mean shifts under the alternative, that is $H(\cdot)={\bf I}_p$ ($p\times p$ identity matrix). They adopt the idea of self-normalization (\cite{shao2010self},~\cite{shao2010testing}) in forming their test statistic and the  theoretical validity of their SN-based test is shown under homoskedasticity.
When there is time-varying heteroscedasticity, we show that the asymptotic null distribution of the SN-based test statistic in \cite{wang2019inference} is no longer pivotal, and it depends on the unknown $H(\cdot)$. To accommodate the unknown heteroscedascity, we propose to use the wild bootstrap to directly approximate the finite sample distribution of the original class of U-statistics, instead of doing self-normalization. With the aid of the recently proposed unconditional approach in justifying bootstrap consistency [\cite{bucher2019note}], we are able to show the consistency of wild bootstrap under the framework (\ref{model1}) and derive the local asymptotic power
under the one-change point alternative. 
In the context of testing for one change point in mean, our bootstrap-based test is free of tuning parameters, and performs well for a broad range of heteroscedastic models in our simulation studies. Extensions to testing for multiple change-point alternative and estimation of change-points using WBS (wild binary segmentation, \cite{fryzlewicz2014wild}) are also made. \change{ Note that like \cite{wang2019inference}, our bootstrap-based test targets  dense alternatives (i.e.when small changes occur for a substantial portion of the components),  which  can be well
motivated by real data and is often the type of alternative we are interested in. For example, copy
number variations in cancer cells are commonly manifested as change-points occurring at the same
positions across many related data sequences corresponding to cancer samples and biologically-
related individuals; see \cite{fan2017}.}

The rest of the paper is structured as follows. Section~\ref{sec:test} describes the test statistic and wild bootstrap scheme for testing a single change point. An extension to testing multiple change points is also made. Section~\ref{sec:theory} provides the assumptions and  theoretical results for the proposed testing procedure under the null and alternatives. In Section~\ref{sec:estimation}, we combine the WBS with our bootstrap-based test for change-point estimation. Section~\ref{sec:simulation}  compares the bootstrap-based  testing and estimation methods with their counterparts in  \cite{wang2019inference} via  simulations. Section~\ref{sec:data} illustrates the usefulness of our method using a real dataset and Section~\ref{sec:conclusion} concludes. All technical details and proofs are relegated to the appendix.
 
\section{Test statistics and bootstrap calibration}

\label{sec:test}
\subsection{Single change point testing}
\label{sub:single}
We first focus on the single change point alternative
\[H_{11}:  \mu_1 = \mu_2 \ldots = \mu_{k_1} \ne  \mu_{k_1 + 1} = \ldots \mu_n. \]
Our test statistic is motivated by \cite{wang2019inference}, which was inspired by the two sample testing statistics in \cite{chen2010two}. \change{For readers who are not familiar with those papers, we now provide a brief introduction to the main ideas which appeared in there. More precisely, suppose $(U_2, V_2)$ is an independent copy of $(U_1,V_1)$. Consider the function
\[
h\{(U_1,V_1),(U_2,V_2)\} = (U_1- V_1)^T(U_2-V_2).
\]
The expectation of this kernel function is 
\[
E[h\{(U_1,V_1),(U_2,V_2)\}] = \|E(U_1) - E(V_1)\|^2.
\]
Note that this expectation equals zero if and only if $E[U_1] = E[V_1]$. A natural unbiased estimator for $E[h\{(U_1,V_1),(U_2,V_2)\}]$ given two independent samples $U_1,\dots,U_n$, $V_1,\dots, V_m$ take sthe form 
\begin{align*}
&\frac{1}{n(n-1)m(m-1)} \sum_{i_1 \neq i_2, i_1,i_2 =1 }^n \sum_{j_1 \neq j_2, j_1,j_2 =1 }^m h((U_{i_1},V_{j_1}),(U_{i_2},V_{j_2}))
\\
= & \frac{4}{n(n-1)m(m-1)} \sum_{1 \le i_1 <i_2 \le n} \sum_{1 \le j_1 < j_2 \le m } h((U_{i_1},V_{j_1}),(U_{i_2},V_{j_2})). 
\end{align*}
Note that this is simply a two-sample U-Statistic with kernel $h$. This statistic was proposed by \cite{chen2010two} for comparing the means of two possibly high-dimensional vectors. The key observation of \cite{chen2010two} was that this statistic is more appropriate than the seemingly natural alternative $\|\bar U - \bar V\|_2^2$ (with $\bar U, \bar V$ denoting the corresponding sample means) because the latter contains terms of the form $(U_i-V_j)^T(U_i-V_j)$ which do not have expected value zero under the null of equal means. This does not matter in fixed dimensions,  but can blow up if the dimension of the vectors grows with sample size. 

Suppose the change in mean vector occurs at time $ k+1$. We can view $X_1,\ldots,X_{k}$ and $X_{k+1},\ldots,X_{n}$ as two independent samples with different means. A natural test statistic for a change at time $k$ is thus 
\begin{align*}
G_n(k) =& \frac{2}{k(k-1)}\frac{2}{(n-k)(n-k-1)}\sum_{1\le i_1<j_1\le k}\sum_{k+1\le i_2<j_2\le n} (X_{i_1}-X_{i_2})^T(X_{j_1}-X_{j_2}) 
\\
= & \frac{2}{k(k-1)}\sum_{1\le i< j \le k}X_i^TX_j + \frac{2}{(n-k)(n-k-1)} \sum_{k+1 \le i< j\le n}X_{i}^TX_{j}
 - \frac{2}{k(n-k)}\sum_{i=1}^{k}\sum_{j=k+1}^{n} X_i^TX_j.
\end{align*}
Here, the second equality follows after straightforward computations and facilitates the theoretical analysis of our test statistic. Since the location $k$ of the change point is unknown, we consider the maximum value over all possible change-points. }



To mimic the CUSUM process used in the low dimensional setting, we define a rescaled version of $G_n(m)$,
\[\Tilde{G}_n(m)  =  \frac{m(m-1)(n-m)(n-m-1)}{n^3} G_n(m),\]
where the rescaling is adopted to prevent the statistics $G_n(m)$ on the two ends  from blowing up. Note that this was implicitly done in the SN-based test statistic of \cite{wang2019inference}. Then we define our test statistic for $H_{11}$ to be 
\[T_n = \max_{m =2 ,3,\ldots,n-2}\Tilde{G}_n(m).\]
This formulation is similar to \cite{wang2019inference} but does not require the use of self-normalization technique, which has its origin from \cite{shao2010self} and \cite{shao2010testing}. \change{Under the null we have $E[G_n(m)] = 0$ for all $n,m$. Hence the statistic $T_n$ is expected to converge to a non-degenerate distribution upon suitable standardization. Under the  single change-point alternative with change at $k_0$ we have $E[G_n(k_0)] > 0$ with magnitude depending on $k_0$ and the size of the change. Hence the test statistic with the same normalization as under the null diverges under the one change-point alternative if the magnitude of change is large enough. 
}
As will be shown later, the limiting null distribution of $T_n$ depends on the unknown $H(\cdot)$, thus is not asymptotically pivotal and the idea of self-normalization is not directly applicable. This motivates us to propose a bootstrap-based approach to approximate the finite sample distribution (or the limiting null distribution) of $T_n$ under the null. 

Specifically, we employ the Gaussian multiplier bootstrap. Let $e_1,\ldots,e_n$ be i.i.d $N(0,1)$ random variables independent of $X_1,X_2,\ldots,X_n$. Let $\bar X = \frac{1}{n} \sum_{i = 1}^n X_i$ denote the sample mean. The bootstrap test statistic is defined as
\[T_n^* = \max_{m =2 ,3,\ldots,n-2}\Tilde{G}^*_n(m),\]
where
\[\Tilde{G}^*_n(m)  =  \frac{m(m-1)(n-m)(n-m-1)}{n^3} G^*_n(m),\]
and
\begin{align*}
    G^*_n(m) =& 
    \frac{2}{m(m-1)}\sum_{1\le i< j \le m}(X_{i}-\bar{X})^T(X_{j}- \bar{X}) e_i e_j
    \\
    & + \frac{2}{(n-m)(n-m-1)}  \sum_{m+1 \le i< j\le n}(X_{i}-\bar{X})^T(X_{j}- \bar{X})e_i e_j\\
    & - \frac{2}{m(n-m)}\sum_{i=1}^{m}\sum_{j=m+1}^{n} (X_{i}-\bar{X})^T(X_{j}- \bar{X})e_i e_j, 
\end{align*}
To ensure the bootstrap consistency, the observations are centered with the overall mean. \change{In practice, we also tried the centering by local mean (e.g., replace $\bar{X}$ by $\frac{1}{m} \sum_{i=1}^{m} X_{i}$ in the first summand of $G_{n}^{*}(m)$), and the results are similar to the ones we obtain by centering by overall mean. The proof and implementation for the latter seem a bit simpler, so we only present the latter.}

\change{In the low dimensional setting, i.e., when $p$ is fixed, the weighted bootstrap for degenerate U-statistic has been studied by \cite{Huskova92}, \cite{Janssen94}, \cite{dehling94}, \cite{wang04}, among others. 
We refer the reader to a recent paper by \cite{bootU2021} and more references therein. We are not aware of any results on bootstrap consistency for degenerate U-Statistics for data of increasing dimension. }

The theoretical bootstrap critical value for a size $\alpha$ test is defined to be 
\[
c_{1,\alpha} = \inf\{t \in \mathbb{R} : P(T_n^* > t|\bm X) \le \alpha\},
\]
where $\bm{X} = (X_1, \ldots, X_n)$. In practice this theoretical value is typically approximated by Monte Carlo simulations. Let $F_M^*$ denote the empirical cdf of $M$ bootstrap statistics $T_n^{*,1},\dots,T_n^{*,M}$, where each of them is based on an independent sequence of multipliers. Then we define 
\[
c_{1,\alpha}^{(M)} = \inf\{t \in \mathbb{R} : 1-F_M^*(t) \le \alpha\}.
\] 
This quantity can be computed through simulations. We reject the null hypothesis when $T_n > c_{1,\alpha}^{(M)} $.

\subsection{Multiple change points testing}
In practice, the number of change points is often unknown, so we consider a more general multiple change-points alternative, 
\change{
$$
H_{1}: \exists s \in \mathbb{N} \text { and } 1<k_{1}<\cdots<k_{s}< n \text { such that }
$$
$$
\mu_{1}=\cdots=\mu_{k_{1}} \neq \mu_{k_{1}+1}=\ldots \cdots \ldots=\mu_{k_{s}} \neq \mu_{k_{s}+1}=\cdots=\mu_{n} .
$$
}
Inspired by the scanning approach developed by \cite{lavitas2018time} for change-point testing in the univariate time series setting, we can incorporate the idea of forward and backward scanning into our test statistics for multiple change points detection. \\
To this end, we first introduce some more general notations. For any $a \le m \le b$, $a,b,m \in \{1,\dots,n\}$ define 
\begin{align*}
    G_n(m;a,b) &= 
    \begin{pmatrix}
    m-a+1\\2
    \end{pmatrix}^{-1}\sum_{a\le i< j \le m}X_i^TX_j + \begin{pmatrix}
    b-m\\2
    \end{pmatrix}^{-1} \sum_{m+1 \le i< j\le b}X_{i}^TX_{j}\\
    & \quad - \frac{2}{(m-a+1)(b-m)}\sum_{i=a}^{m}\sum_{j=m+1}^{b} X_i^TX_j,
    \\
\Tilde{G}_n(m;a,b) & =  \frac{(m-a+1)(m-a)(b-m)(b-m-1)}{(b-a+1)^3} G_n(m;a,b)
\end{align*}
It is obvious that $G_n(m) = G_n(m;1,n)$ and $\Tilde{G}_n(m) = \Tilde{G}_n(m;1,n).$
Our test statistic for multiple change points alternative takes the following form
\[T_{n,M} = \max_{1 \le m < k \le n } \Tilde{G}_n(m;1,k) + \max_{1\le k<m \le n }  \Tilde{G}_n(m;k,n).\]
Under $H_0$, since there is no change point, both forward and backward scanning parts are expected to be small. When there is at least one change point, the first change point would result in an inflation of the forward scanning part and the last change point would lead to a large value for  the backward scanning part.

Again, we use the Gaussian multiplier bootstrap to obtain the bootstrap distribution and calibrate the size. The bootstrap statistic is defined as 
\[T^*_{n,M} = \max_{1 \le m < k\le n} \Tilde{G}_n^*(m;1,k) + \max_{1\le k<m \le n}  \Tilde{G}_n^*(m;k,n),\]
where 
\[\Tilde{G}_n^*(m;a,b)  =  \frac{(m-a+1)(m-a)(b-m)(b-m-1)}{(b-a+1)^3} G_n^*(m;a,b),\]
and 
\begin{align*}
    G_n^*(m;a,b) = &
    \begin{pmatrix}
    m-a+1\\2
    \end{pmatrix}^{-1}\sum_{a\le i< j \le m}(X_i-\bar X)^T(X_j - \bar X) e_i e_j
    \\
    & + \begin{pmatrix}
    b-m\\2
    \end{pmatrix}^{-1} \sum_{m+1 \le i< j\le b}(X_i-\bar X)^T(X_j - \bar X)e_i e_j\\
    & - \frac{2}{(m-a+1)(b-m)}\sum_{i=a}^{m}\sum_{j=m+1}^{b} (X_i-\bar X)^T(X_j - \bar X)e_i e_j.
\end{align*}
The bootstrap critical value is defined to be 
\[c_{2,\alpha} = \inf\{t \in \mathbb{R} : P(T_{n,M}^* > t|\bm X) \le \alpha\}.\]
In practice, the critical value is approximated by $c_{2,\alpha}^{(M_n)}$, which is computed from the $M_n$ bootstrap samples, similarly as  $c_{1,\alpha}^{(M_n)}$. We then reject the null hypothesis when $T_{n,M} > c_{2,\alpha}^{(M_n)} $. It is worth noting that the proposed bootstrap test avoids the trimming parameter that is required in \cite{lavitas2018time} and \cite{wang2019inference}, and is thus tuning parameter free. 


\section{Theoretical results}
\label{sec:theory}
In this section, we present the theoretical results regarding the asymptotic properties of the test statistics and bootstrap consistency. 
\csv{Throughout this paper, we work with triangular array asymptotics where $Z_1,\dots,Z_n$ are independent across $n$ but with dimension $p = p_n$ that can grow with $n$. In order to keep the notation simple, we will not explicitly mark the dependence of the distribution and dimension of $Z$ on $n$. All asymptotics will be for $n \to \infty$. }
\change{For a symmetric matrix $\Sigma$, we denote $\|\Sigma\|_F$ its Frobenius norm. }
Consider the model~(\ref{model1}), 
where $Z_i$'s are i.i.d. $p-$dimensional random vectors with $E[Z_1] = 0$, $E[Z_1Z_1^T] = \Sigma$. The main technical assumptions are displayed below.

\begin{assumption}
\label{as:1}
 $tr(\Sigma^4) =o(\|\Sigma\|^4_F)$.
\end{assumption}
\begin{assumption}
\label{as:2}
$\sum_{l_1,\ldots,l_h=1}^p cum^2(Z_{1,l_1},\ldots,Z_{1,l_h}) \le C||\Sigma||^h_F$ for $h=2,3,4,5,6$ and some positive constant $C$ \change{which does not depend on $n$}.
\end{assumption}
\begin{assumption}
\label{as:3}
For every $t\in [0,1]$, $H(t)$ is a $p\times p$ diagonal matrix with all diagonal elements bounded by some finite constant $B$, \change{independent of $n$} that is 
\[|H_{l,l}(t) | \le B \text{ for all $t \in [0,1]$ and \change{$l = 1, \ldots, p_n, n \geq 1$}}.\]
\end{assumption}
\begin{assumption}
\label{as:4} 
\csv{Assume that the following limit 
\[
V(a,b) :=\lim_{n \to \infty} \frac{1}{n^2||\Sigma||_F^2} \sum_{i=\floor{na}+1}^{\floor{nb}-1}  \sum_{j =  \floor{na}+1}^i \tr\Big( H^2\Big(\frac{j}{n}\Big)H^2\Big(\frac{i+1}{n}\Big)\Sigma^2\Big)
\]
exists for all $0\le a \le b\le 1$.
}
\end{assumption}

Assumptions~\ref{as:1} and \ref{as:2} are also imposed in \cite{wang2019inference}.
 \change{As shown in \cite{wang2019inference}, Assumption~\ref{as:1} is equivalent to $\|\Sigma\|_2=o(\|\Sigma\|_F)$ and can only hold when $p = p_n \rightarrow \infty$ as $n\to \infty$. See page 813 of \cite{chen2010two} for additional discussion on its implications on the eigenvalues of $\Sigma$. Assumption~\ref{as:2} was first imposed in \cite{wang2019inference}, which is shown to be weaker than the factor-model-like assumption in \cite{chen2010two}; see Remark 3.3 in \cite{wang2019inference}. The summability of cumulants assumption is commonly used in time series analysis for the asymptotic analysis of low-dimensional time series [\cite{Brillinger1975}]. In our setting, it is used to ensure that the dependence is weak enough across the dimension of the vector. This is a crucial technical ingredient in our asymptotic analysis when establishing finite-dimensional convergence of a suitably normalized version of the process $\widetilde{G}_n$ to a multivariate normal limit.} \change{ Assumption~\ref{as:2} in general holds under uniform bounds on moments and `short-range’ dependence  conditions on the components of $X$ (possible after permutation). For example, if the sequence corresponding to the ordered components of $X$ (or a permutation of components) satisfies certain mixing and moment conditions, then Assumption~\ref{as:2} holds. See Remark 3.2 of \cite{wang2019inference} for more discussion and references. }
 

Assumptions~\ref{as:3} and ~\ref{as:4} are regarding the time varying heteroskedasticity $\{H(t)\}$. Assumption~\ref{as:3} bounds the range of heteroskedasticity, and is mild. 
\csv{Assumption~\ref{as:4} appears when we study the process $t \mapsto \tilde G_n(\floor{nt})$. More precisely, in the Appendix we decompose $\tilde G_n(\floor{nt})$ into a linear combination of a process $\tilde S_n$ evaluated at different points (see beginning of the proof of Theorem~\ref{thm:1}). The limiting variance of this process $\tilde S_n$ is directly related to the limit appearing in Assumption~\ref{as:4} (see the proof of Proposition~\ref{prop:convSn}). For a transparent example, assume $H(t) = f(t) I_p$ where $f$ is a real-valued function and $I_p$ denotes the $p\times p$ identity matrix. In this case the assumption simplifies because $tr(H^2(i/n)H^2(j/n)\Sigma^2)$ reduces to $f^2(i/n)f^2(j/n) tr(\Sigma^2) = f^2(i/n)f^2(j/n) \|\Sigma\|_F^2$. The normalized sum can be seen as a Riemann approximation of an integral, and convergence takes place provided that $f$ is sufficiently regular (for instance, bounded and piece-wise continuous with a finite number of jumps.) In general settings, Assumption~\ref{as:4} boils down to requiring sufficient regularity of each component of $H$ in a suitable uniform sense. In what follows, we use $I(\cdot)$ to denote the indicator function. }

\begin{theorem}\label{thm:1}
Under Assumptions~\ref{as:1}-\ref{as:4} and $H_0$, the normalized test statistic $T_n$ converges to a non-pivotal distribution after proper standardization, that is, 
\[\frac{T_n}{\|\Sigma\|_F} \xrightarrow{d} T = \sup_{r \in [0,1]} G(r),\]
where 
\[G(r) := 2(1-r)Q(0,r) + 2rQ(r, 1) - 2r(1-r) Q(0, 1),\]
and $Q$ is a mean-zero  Gaussian process on $[0,1]^2$, and the covariance is given by 
\[Cov(Q(a_1,b_1), Q(a_2,b_2)) = V(a_1\vee a_2, b_1\wedge b_2) I(b_1\wedge b_2 > a_1 \vee a_2).\]
\end{theorem}
The theorem implies that the normalized test statistic converges to a potentially non pivotal distribution which depends on the time varying heteroskedasticity function $H(\cdot)$. When the time varying heteroskedasticity function is the identity matrix (i.e., $H(t) = I_p$ for every $t\in [0,1]$), the covariance structure of $Q$ is the same as that in Theorem 3.4 of \cite{wang2019inference} and the limit is pivotal. Self-normlization can then help to get rid of the unknown normalizing factor $\|\Sigma\|_F$ leading to a pivotal test. However, in general the distribution of the self normalized statistic from \cite{wang2019inference} is not pivotal due to presence of the unknown heteroskedasticity function $H(\cdot)$ in the definition of $V$.

Next we present the results on the bootstrap consistency under $H_0$. Additional assumptions are needed to establish bootstrap consistency. In particular, we assume 
\begin{assumption}
\label{as:5}
Assume that $tr(\Sigma)^2 = o(n^2\|\Sigma\|_F^2)$ and 
\[
\frac{\sum_{s,t =1}^p cum(Z_{1,s}, Z_{1,s}, Z_{1,t}, Z_{1,t})}{n^2\|\Sigma\|_F^2} \to 0 .
\]
\end{assumption}
As shown in the Appendix, Assumption~\ref{as:5} implies \[\kappa_4 := E[Z_1^TZ_1Z_1^TZ_1] =  o(n^2||\Sigma||_F^2),\]
which is comparable to Assumption~\ref{as:2} and can be verified under similar weak dependence structure as described in \cite{wang2019inference}. This assumption is used when showing the negligibility of some remainder terms for the bootstrap process.  

\begin{theorem}\label{thm:2}
Assume Assumptions~\ref{as:1}-\ref{as:5} hold. Under $H_0$, we have for any sequence $M_n \to \infty$ and any $\alpha < 1/2$: $P(T_n > c_{1,\alpha}^{(M_n)}) \to \alpha$.

\end{theorem}

Next we state the result regarding the power of the proposed test statistics. 
\begin{theorem}\label{thm:3}
\change{
Suppose that Assumptions~\ref{as:1}-\ref{as:5} hold. }
Assume there is one single change point at $k_1 := \floor{nc}$, $\mu_i =\mu ,i=1,..., k_1$ and $\mu_i = \mu + \Delta,i= k_1 + 1 \ldots, n.$ Then for any sequence $M_n \to \infty$ and any $\alpha < 1/2$ 
\begin{enumerate}
\item(Diminishing local alternative) If $\frac{n\|\Delta\|_2^2}{\|\Sigma\|_F} \to 0$, 
then $P(T_n > c_{1,\alpha}^{(M_n)}) \to \alpha$.
\item(Diverging local alternative) If $\frac{n\|\Delta\|_2^2}{\|\Sigma\|_F} \to \infty$, 
then $P(T_n > c_{1,\alpha}^{(M_n)}) \to 1$.
\item(Fixed local alternative) If $\frac{n\|\Delta\|_2^2}{\|\Sigma\|_F} \to \beta \in (0, \infty)$,  \[\frac{T_n}{\|\Sigma\|_F} \xrightarrow{d} \sup_{r \in [0,1]} \left(G(r;0,1) + \Lambda(r)\right), \]
where \[\Lambda(r) = \begin{cases} (1-c)^2r^2\beta & r \le c\\
c^2 (1-r)^2 \beta & r > c
\end{cases}.
\]
Moreover, for \change{$2$ copies of the bootstrap statistic $T_{n}^{1,*},T_{n}^{2,*}$} which are based on independent sets of multipliers we have 
\[
\change{(T_{n}/\|\Sigma\|_F,T_{n}^{1,*}/\|\Sigma\|_F,T_{n}^{2,*}/\|\Sigma\|_F) \xrightarrow{d} (T,T^{(1)},T^{(2)})},
\]
where \change{$T,T^{(1)},T^{(2)}$} are independent copies of $T$ from Theorem~\ref{thm:1}.


\end{enumerate}
\end{theorem}
The result in Theorem~\ref{thm:3} shows that our test has nontrivial power when the $\mathcal{L}_2$-norm of change is large relative to $\|\Sigma\|_F$, which targets the dense alternative. In the special homoscedastic case, i.e., $H(t)={\bf I}_p$ for all $t\in [0,1]$, the power result is consistent with the one obtained in \cite{wang2019inference}, in the sense that both tests share the same rate of alternative under which nontrivial power occurs. This suggests that the bootstrap-based procedure brings extra robustness with respect to unconditional time-varying heteroscedasticity, as compared to the SN-based one in \cite{wang2019inference}, without sacrificing power. \change{Finally, we note that by results in \cite{bucher2019note} the result in part 3 remains true for an arbitrary number of bootstrap copies $T_{n}^{*,1},\dots,T_{n}^{*,K}$ that are obtained from independent multipiers.}


The following theoretical results can be derived similarly for multiple change point testing. 
\begin{theorem}\label{thm:4}
Assume Assumptions~\ref{as:1}-\ref{as:5} hold, under $H_0$, we have 
\[\frac{T_{n,M}}{\|\Sigma\|_F} \xrightarrow{d} T_M := \sup_{0\le r_1 < r_2 \le 1} G(r_1;0,r_2) + \sup_{0\le r_1 < r_2 \le 1} G(r_2;r_1,1),\]
where 
\[G(r;a,b) := 2(b-a)(b-r)Q(a,r) + 2(b-a)(r-a)Q(r, b) - 2(r-a)(b-r) Q(a, b).\]
Moreover, \change{for $2$ copies of the bootstrap statistic $T_{n,M}^{1,*},T_{n,M}^{2,*}$} which are based on independent sets of multipliers we have 
\[
\change{(T_{n,M}/\|\Sigma\|_F,T_{n,M}^{1,*}/\|\Sigma\|_F,T_{n,M}^{2,*}/\|\Sigma\|_F) \xrightarrow{d} (T_M, T^{(1)}_M,T^{(2)}_M)},
\]
where $\change{(T_{n,M},T^{(1)}_M,T^{(2)}_M})$ are independent copies of $T_M$.
\end{theorem}
Under the alternative, we show that the power of the proposed test for multiple change-point detection goes to 1 when there is a dense mean change. 
\begin{theorem} \label{th:multCppower}
\change{Assume Assumptions~\ref{as:1}-\ref{as:5} hold.}
Suppose that there are change points at $k_1, \ldots, k_s$, that $k_j = \floor{c_jn}$ for constants $0 < c_1 < \dots<c_s < 1$, and at least one of the change-point sizes, say for the r'th change-point, satisfies $\frac{n\|\Delta_r\|_2^2}{\|\Sigma\|_F} \to \infty$. Then $P(T_{n,M} > c_{2,\alpha}^{(M_n)}) \to 1$.
\end{theorem}


\section{Change-point estimation}
\label{sec:estimation}

Wild binary segmentation (WBS) was introduced by \cite{fryzlewicz2014wild} as an alternative to the popular binary segmentation algorithm to estimate the change-points locations in a univariate sequence. \cite{wang2019inference} combined WBS and their SN-based test and showed that 
WBS outperforms binary segmentation, especially when the changes are non-monotonic.   Here, we shall combine our bootstrap-based test with the WBS algorithm  to estimate the locations of  change-points in the mean of high-dimensional heteroscedastic data. The algorithm involves  generating $N$ random segments $\{(s_{m},e_{m})\}_{m = 1,\ldots,N}$,  calculating the single change point test statistic on each segment $(s_m,e_m)$,  \change{
$$W(s_{m},e_{m}) = \max_{s_m+2 \le t \le e_m-2 } \Tilde{G}_n(t;s_m,e_m),$$}
and then taking a maximum over all random segments, that is, $\max_{m = 1, \ldots, N} W(s_{m},e_{m})$. A change point is detected  when $\max_{m = 1, \ldots, N}  W(s_{m},e_{m}) > \xi_n$, where $\xi_n$ is a proper threshold parameter. \change{In the event that a change point is detected, let $\widehat m = \argmax_m W(s_{m},e_{m}). $ The location of the changepoint is estimated at 
$$\hat{t}_1=argmax_{s_{\widehat m}+2 \le t \le e_{\widehat m}-2 } \Tilde{G}_n(t;s_{\widehat m},e_{\widehat m}).$$
Then the data is split into two parts $(X_1,\cdots,X_{\hat{t}_1})$ and $(X_{\hat{t}_1+1},\cdots,X_n)$ and WBS is employed for each part until no change-points are detected. 
}


In \cite{wang2019inference},  the threshold was obtained by applying the same test to the simulated iid Gaussian data to the same set of random segments. This approach makes intuitive sense since SN-based test statistic is asymptotically pivotal when there is no heterosecasticity,  but is no longer meaningful in the presence of heteroscedasticity, as the asymptotic pivotal nature of the SN-based test statistic is lost and the function $H(\cdot)$ is unknown. To overcome this difficulty, we propose to adopt a bootstrap-based approach in determining the threshold. Specifically, for $N$ random segments $(s_m,e_m)$, we generate $R$ independent copies of Gaussian multipliers. 
Let \[W^{*(i)}(s,e) =\max_{s+2 \le t \le e-2 } \Tilde{G}^{*(i)}_n(t;s,e),\]
be the $i$th bootstrap-based test statistic on the interval $[s,e]$.  For the $i$th bootstrap replicate, we calculate
\[\hat \xi_n^i = \max_{m=1,\cdots,N} W^{*(i)}(s_m,e_m).\]
The threshold $\xi_n$ is defined as the $95\%$ quantile of the values $\{\hat \xi_n^1,\ldots,\hat \xi_n^R\}$. Note that we generate multipliers once for each bootstrap replication and apply the same multipliers in all intervals. 
\change{Changepoints are now estimated by running WBS$(1,n,\xi_n, \emptyset)$ below.} 


\begin{algorithm}[H]
\caption{Bootstrap-based WBS}
\SetAlgoLined
{WBS$(s,e,\xi_n,\hat C)$}\\
Set of estimated changepoints: $\hat C$\\ 
 \If{$e - s < 4$}{
  STOP\;}
  \Else{ $\mathcal{M}_{s,e}$: = set of those $1 \le m \le N$ for which $s \le s_m$, $e_m \le e$\\
  $m_0 :=\argmax_{m \in \mathcal{M}_{s,e}} W(s_m,e_m)$ \\
  \If{$W(s_{m_0},e_{m_0}) > \xi_n$}{
    Add $m_0$ to the set of estimated change-points $\hat C$\;
    WBS$(s,m_0,\xi_n,\hat C)$\;
    WBS$(m_0+1,e, \xi_n,\hat C)$\;
   }
   \Else{Stop\;}
}
\end{algorithm}

\section{Simulation studies}
\label{sec:simulation}

In this section, we investigate  the finite sample performance of our proposed bootstrap-based tests and WBS+Bootstrap estimation method via simulations. In Section~\ref{sub:test}, we present the size and power for our bootstrap-based tests in comparison with SN-based tests in \cite{wang2019inference} for the settings of single and multiple change points in high-dimensional homoskedastic and heteroscedastic data. Section~\ref{sub:estimation} examines the performance of the WBS+Bootstrap change point estimation method in comparison with the WBS+SN based approach in \cite{wang2019inference} when the unconditional heteroscedasticity is present. 

\subsection{Testing}
\label{sub:test}
Recall that we assume the following data generating model 
\[X_i = \mu_i + H(i/n) Z_i, \text{ for $i = 1,\ldots,n$.}\]
We generate $Z_i, i=1,\cdots,n$ independently from a multivariate normal distribution $MVN(\bm 0,\Sigma)$, where the following three different types of covariance matrix $\Sigma$ are considered, 
\begin{itemize}
    \item (Case 1) AR(1) covariance matrix with $\Sigma_{ij} = 0.5^{|i-j|}$;
    \item (Case 2)  AR(1) covariance matrix with $\Sigma_{ij} = 0.8^{|i-j|}$;
    \item (Case 3) Compound symmetric covariance matrix with $\Sigma_{ij} = 0.5^{{\bf 1}(i \ne j)}$.
\end{itemize}
Cases 1 and 2 both belong to weakly dependent (across coordinates of $X$) models and it will be interesting to see how the increased dependence from Case 1 to Case 2 impact the finite sample size accuracy.    Case 3 corresponds to a model with strong dependence, and it violates the componentwise weakly dependent assumption we imposed in our theory (see Assumptions 1\&2). Nevertheless it would be interesting to see how robust our bootstrap-based tests are with respect to strong componentwise dependence. 

Next, we consider  the following time varying trend function $H(\cdot)$, which specifies the trend in time-varying variance of each component but not the trend in mean. We use the terminology "trend" with the understanding that it always refers to the time-varying variance.
\begin{itemize}
    \item A0: $H(i/n) = {\bf I}_p$, $i=1,\cdots,n$. 
    This is the case for no trend.
    \item A1: $H(i/n) = \{0.2\bm 1_{i\le n/2} +  0.6 \bm 1_{i> n/2}\}{\bf I}_p$ (piecewise constant trend). 
    \item A2: $H(i/n) = (i/n){\bf I}_p$ (linear trend). 
    \item A3:  $H(i/n) = [0.2\{1 + \cos^2(i/n^{4/5})\}] {\bf I}_p$. This trend function has a cosine shape.  
    \item A4: $H(i/n) = \{0.2+0.1 \log(1+|i-n/2|)\} {\bf I}_p$. This trend function has a sharp change around $n/2$.
    \item A1 + A2: Apply trend function A1 to the first $p/2$ coordinates in $Z$, and apply trend function A2 to the rest of coordinates. 
    \item A1 + A3: Apply trend function A1 to the first $p/2$ coordinates in $Z$, and apply trend function A3 to the rest of coordinates. 
    \item A1 + A4: Apply trend function A1 to the first $p/2$ coordinates in $Z$, and apply trend function A4 to the rest of coordinates. 
\end{itemize}
Some of these trend functions, such as A1, A3 and A4, have been considered in \cite{zhaoli2013}, who studied the inference of the mean for a univariate time series with time-varying variance.

First, we investigate the case where there is at most one change point in the mean. Under the null hypothesis, we set $\mu_i = \bm 0$ for all $i = 1, \ldots, n$. \change{We consider $(n,p) = (400,100), (100, 100), (400, 400)$}, for all choices of $\Sigma$ and $H(\cdot)$ described above. The empirical sizes  at significance levels $\alpha = 0.05,0.1$ are reported based on $1000$ Monte Carlo simulations. The results of SN-based test statistic for one change point (i.e., $T_n$  in \cite{wang2019inference}) are also reported for comparison. From Table \ref{tab_size}, we can see that for AR covariance matrix with $\rho = 0.5, 0.8$, both tests achieve size accuracy, i.e., the empirical sizes are close to the nominal level, when there is no time-varying  heteroscedasticity. However, when time varying heteroscedasticity is present, the SN-based test exhibits pronounced over-size distortion in the case of A1, A2, A1+A2, and A1+A3. By contrast, the bootstrap-based test we propose achieves accurate size across all trend types. When the covariance matrix is compound symmetric, the model assumptions required for the validity of both SN-based test and bootstrap-based test are violated. It is observed that the SN-based test over-rejects even when there is no  trend, which is consistent with the result in \cite{wang2019inference}. Interestingly, the bootstrap-based test still maintains accurate size for all settings. This suggests that the applicability of bootstrap-based test may be broader than what we are able to justify. It would be interesting but may be challenging to provide a new theory that supports the robustness of our bootstrap-based test when the panel dependence is strong. 



 \begin{table}[]
         \centering

\begin{tabular}{lrrrrrrrrrrrr}
\hline

\multicolumn{1}{c}{$n = 400$} & \multicolumn{4}{c}{AR 0.5}                                                                              & \multicolumn{4}{c}{AR 0.8}                                                                              & \multicolumn{4}{c}{CS}                                                                                  \\ \cline{2-13} 
\multicolumn{1}{c}{$p=100$}   & \multicolumn{2}{c}{SN}                             & \multicolumn{2}{c}{Boot}                           & \multicolumn{2}{c}{SN}                             & \multicolumn{2}{c}{Boot}                           & \multicolumn{2}{c}{SN}                             & \multicolumn{2}{c}{Boot}                           \\ \hline

$\alpha$  & \multicolumn{1}{c}{0.05} & \multicolumn{1}{c}{0.1} & \multicolumn{1}{c}{0.05} & \multicolumn{1}{c}{0.1} & \multicolumn{1}{c}{0.05} & \multicolumn{1}{c}{0.1} & \multicolumn{1}{c}{0.05} & \multicolumn{1}{c}{0.1} & \multicolumn{1}{c}{0.05} & \multicolumn{1}{c}{0.1} & \multicolumn{1}{c}{0.05} & \multicolumn{1}{c}{0.1} \\ \hline
A0        & 0.050                    & 0.096                   & 0.052                    & 0.106                   & 0.063                    & 0.090                   & 0.05                     & 0.098                   & 0.093                    & 0.123                   & 0.046                    & 0.093                   \\
A1        & 0.171                    & 0.284                   & 0.050                    & 0.097                   & 0.181                    & 0.268                   & 0.057                    & 0.099                   & 0.141                    & 0.176                   & 0.053                    & 0.093                   \\
A2        & 0.244                    & 0.341                   & 0.049                    & 0.113                   & 0.223                    & 0.316                   & 0.052                    & 0.104                   & 0.151                    & 0.194                   & 0.046                    & 0.101                   \\
A3        & 0.041                    & 0.074                   & 0.052                    & 0.109                   & 0.043                    & 0.077                   & 0.061                    & 0.106                   & 0.090                    & 0.113                   & 0.049                    & 0.094                   \\
A4        & 0.038                    & 0.08                    & 0.045                    & 0.100                   & 0.052                    & 0.090                   & 0.045                    & 0.102                   & 0.099                    & 0.124                   & 0.048                    & 0.094                   \\
A1+A2     & 0.217                    & 0.313                   & 0.051                    & 0.111                   & 0.193                    & 0.298                   & 0.059                    & 0.106                   & 0.150                    & 0.179                   & 0.06                     & 0.097                   \\
A1+A3     & 0.126                    & 0.198                   & 0.052                    & 0.106                   & 0.12                     & 0.188                   & 0.051                    & 0.096                   & 0.133                    & 0.169                   & 0.05                     & 0.092                   \\
A1+A4     & 0.054                    & 0.090                   & 0.056                    & 0.103                   & 0.056                    & 0.095                   & 0.045                    & 0.098                   & 0.099                    & 0.134                   & 0.053                    & 0.092                   \\ \hline
\end{tabular}
         
\begin{tabular}{lrrrrrrrrrrrr}
\hline
\multicolumn{1}{c}{$n = 100$} & \multicolumn{4}{c}{AR 0.5}                                                                              & \multicolumn{4}{c}{AR 0.8}                                                                              & \multicolumn{4}{c}{CS}                                                                                  \\ \cline{2-13} 
\multicolumn{1}{c}{$p = 100$} & \multicolumn{2}{c}{SN}                             & \multicolumn{2}{c}{Boot}                           & \multicolumn{2}{c}{SN}                             & \multicolumn{2}{c}{Boot}                           & \multicolumn{2}{c}{SN}                             & \multicolumn{2}{c}{Boot}                           \\ \hline
\multicolumn{1}{c}{$\alpha$}  & \multicolumn{1}{c}{0.05} & \multicolumn{1}{c}{0.1} & \multicolumn{1}{c}{0.05} & \multicolumn{1}{c}{0.1} & \multicolumn{1}{c}{0.05} & \multicolumn{1}{c}{0.1} & \multicolumn{1}{c}{0.05} & \multicolumn{1}{c}{0.1} & \multicolumn{1}{c}{0.05} & \multicolumn{1}{c}{0.1} & \multicolumn{1}{c}{0.05} & \multicolumn{1}{c}{0.1} \\ \hline
A0                            & 0.047                    & 0.085                   & 0.049                    & 0.106                   & 0.055                    & 0.097                   & 0.043                    & 0.085                   & 0.116                    & 0.140                   & 0.068                    & 0.128                   \\
A1                            & 0.158                    & 0.238                   & 0.044                    & 0.098                   & 0.165                    & 0.228                   & 0.043                    & 0.109                   & 0.133                    & 0.170                   & 0.056                    & 0.111                   \\
A2                            & 0.208                    & 0.286                   & 0.046                    & 0.106                   & 0.205                    & 0.290                   & 0.046                    & 0.108                   & 0.157                    & 0.207                   & 0.057                    & 0.118                   \\
A3                            & 0.052                    & 0.094                   & 0.048                    & 0.114                   & 0.056                    & 0.088                   & 0.050                    & 0.096                   & 0.117                    & 0.147                   & 0.071                    & 0.132                   \\
A4                            & 0.036                    & 0.068                   & 0.047                    & 0.106                   & 0.049                    & 0.081                   & 0.040                    & 0.090                   & 0.111                    & 0.146                   & 0.060                    & 0.128                   \\
A1+A2                         & 0.173                    & 0.263                   & 0.048                    & 0.112                   & 0.185                    & 0.256                   & 0.045                    & 0.108                   & 0.156                    & 0.196                   & 0.058                    & 0.115                   \\
A1+A3                         & 0.078                    & 0.134                   & 0.044                    & 0.100                   & 0.087                    & 0.130                   & 0.050                    & 0.111                   & 0.108                    & 0.135                   & 0.060                    & 0.120                   \\
A1+A4                         & 0.041                    & 0.083                   & 0.042                    & 0.100                   & 0.063                    & 0.103                   & 0.039                    & 0.098                   & 0.106                    & 0.142                   & 0.058                    & 0.123                   \\ \hline
\end{tabular}         
         
\begin{tabular}{lrrrrrrrrrrrr}
\hline
\multicolumn{1}{c}{$n = 400$} & \multicolumn{4}{c}{AR 0.5}                                                                              & \multicolumn{4}{c}{AR 0.8}                                                                              & \multicolumn{4}{c}{CS}                                                                                  \\ \cline{2-13} 
\multicolumn{1}{c}{$p=400$}   & \multicolumn{2}{c}{SN}                             & \multicolumn{2}{c}{Boot}                           & \multicolumn{2}{c}{SN}                             & \multicolumn{2}{c}{Boot}                           & \multicolumn{2}{c}{SN}                             & \multicolumn{2}{c}{Boot}                           \\ \hline
\multicolumn{1}{c}{$\alpha$}  & \multicolumn{1}{c}{0.05} & \multicolumn{1}{c}{0.1} & \multicolumn{1}{c}{0.05} & \multicolumn{1}{c}{0.1} & \multicolumn{1}{c}{0.05} & \multicolumn{1}{c}{0.1} & \multicolumn{1}{c}{0.05} & \multicolumn{1}{c}{0.1} & \multicolumn{1}{c}{0.05} & \multicolumn{1}{c}{0.1} & \multicolumn{1}{c}{0.05} & \multicolumn{1}{c}{0.1} \\ \hline
A0                            & 0.051                    & 0.101                   & 0.059                    & 0.124                   & 0.043                    & 0.088                   & 0.047                    & 0.115                   & 0.094                    & 0.128                   & 0.050                     & 0.090                    \\
A1                            & 0.186                    & 0.284                   & 0.057                    & 0.105                   & 0.180                     & 0.264                   & 0.049                    & 0.099                   & 0.147                    & 0.179                   & 0.052                    & 0.089                   \\
A2                            & 0.241                    & 0.344                   & 0.058                    & 0.115                   & 0.255                    & 0.350                    & 0.052                    & 0.113                   & 0.151                    & 0.189                   & 0.049                    & 0.098                   \\
A3                            & 0.035                    & 0.073                   & 0.057                    & 0.124                   & 0.040                     & 0.069                   & 0.056                    & 0.102                   & 0.090                     & 0.115                   & 0.051                    & 0.092                   \\
A4                            & 0.046                    & 0.087                   & 0.061                    & 0.120                    & 0.039                    & 0.073                   & 0.047                    & 0.096                   & 0.100                      & 0.129                   & 0.052                    & 0.092                   \\
A1+A2                         & 0.223                    & 0.326                   & 0.062                    & 0.117                   & 0.220                     & 0.322                   & 0.061                    & 0.114                   & 0.148                    & 0.174                   & 0.053                    & 0.100                     \\
A1+A3                         & 0.118                    & 0.192                   & 0.052                    & 0.098                   & 0.123                    & 0.193                   & 0.056                    & 0.110                    & 0.130                     & 0.166                   & 0.051                    & 0.090                    \\
A1+A4                         & 0.047                    & 0.092                   & 0.060                     & 0.114                   & 0.040                     & 0.084                   & 0.043                    & 0.098                   & 0.105                    & 0.135                   & 0.050                     & 0.090                    \\ \hline
\end{tabular}
        \caption{Size for single change point testing under different trend functions}
        \label{tab_size}
    \end{table}

Next, we investigate the power of the proposed bootstrap test under the alternative of one change point. We consider $(n,p) = (100,100)$ and the mean shift occurs at the center of data, i.e. $\mu_i = \Delta \bm 1\{i \ge \floor{n/2}\}$. We provide the power curves of the proposed bootstrap test and SN-based test for two AR covariance matrices and all trend types. We let $\Delta$ steadily increase from $0$ to some larger values and evaluate the empirical power at different change magnitudes based on $1000$ Monte Carlo simulations. In Figure~\ref{fig:power_curve}, the solid line corresponds to the power for bootstrap-based test and the dashed line corresponds to  SN-based test, with the colors red and black indicating the results for $\rho=0.8$ and $\rho=0.5$, respectively.  When there is no time varying trend in variance, the two tests have similar size and power. Similar phenomenon holds for trend types A3, A4, and A1+A4, for  all of which we observe size accuracy for both tests. On the other hand, the SN-based test  shows significant size distortion for trend types A1, A2, A1+A2 and A1+A3, making it difficult to compare the power of the two tests directly. To make a fair comparison, we also report the size adjusted power of the SN-based test for these cases. To be more specific, we calibrate the empirical critical values used in SN-based test such that the empirical sizes are exactly $0.05$. The size adjusted powers of SN-based test are shown in dotted lines in the figures for trend types A1, A2, A1+A2, and A1+A3. The size for the bootstrap-based test is fairly close to 0.05, so  we did not make any  power adjustment. A direct comparison between the size-adjusted power of SN-based test and the raw power of bootstrap-based test suggests that the powers are quite comparable, with slight advantage for the bootstrap-based test in some settings, such as A1, A2, and A1+A2.


\begin{figure}
    \centering
    \includegraphics[scale = 0.9]{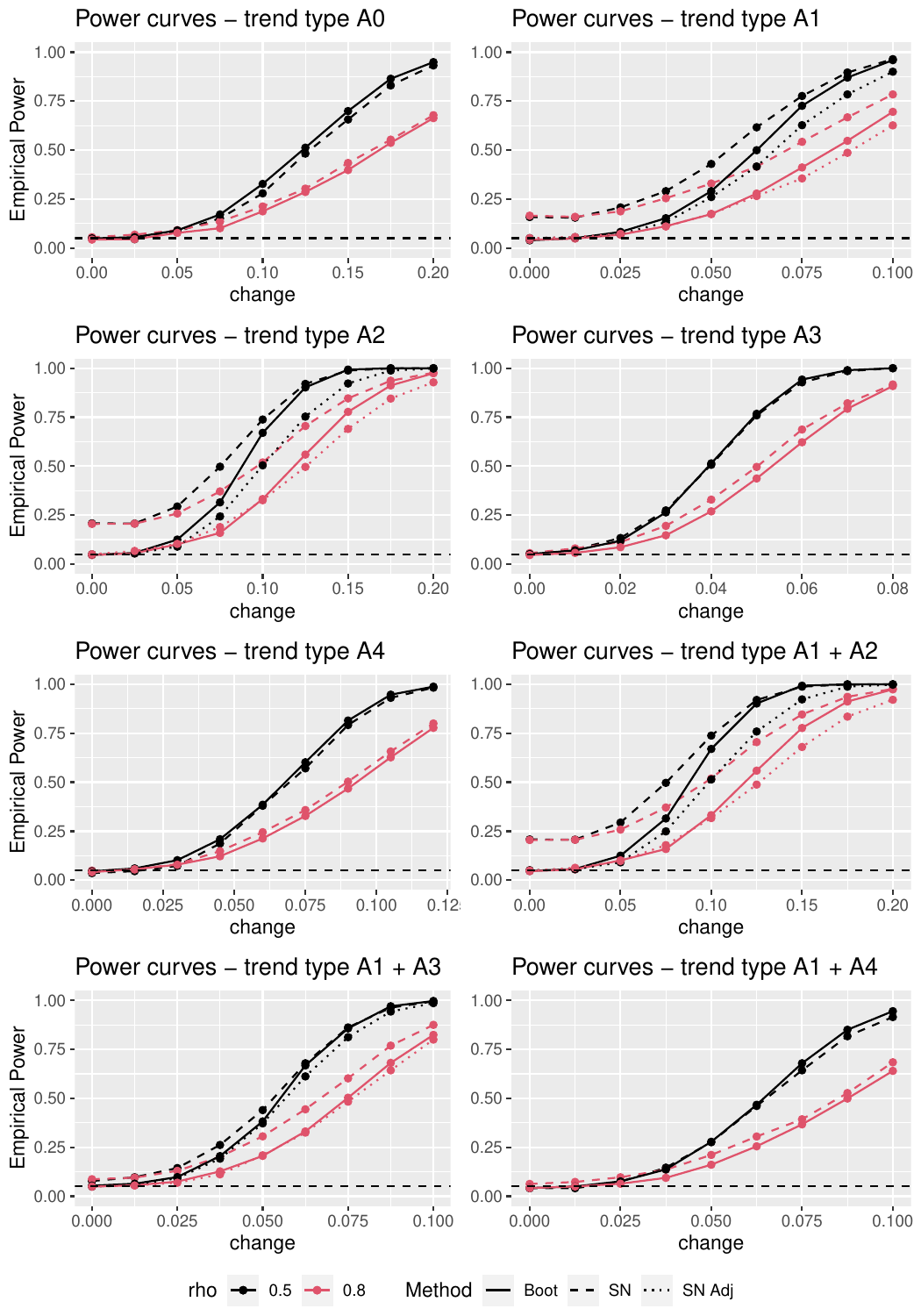}
    \caption{Power curves for single change point testing under different trend functions}
    \label{fig:power_curve}
\end{figure}

Next, we investigate the performance of the bootstrap-based test that targets unknown number of change points, where there are more than one change point under the alternative.  We only present the results for trend types A0, A1, A2 and A1 + A2.  Following \cite{wang2019inference}, we consider  a two-change-points alternative (2CP)
\[\mu_i =  \Delta \bm 1\{ \floor{n/3} \le i \le \floor{2n/3}\} \]
and  a three-change-points alternative (3CP),
\[\mu_i  = \Delta \bm 1\{ \floor{n/4} \le i \le \floor{n/2} \} +  \Delta \bm 1\{ \floor{3n/4}\le i \le n\}.
\]
We consider two AR covariance matrices used before when generating $Z_i$ and set $(n,p) = (50,50)$ and $\Delta = 0.2$.   We  compare the empirical size and power  with those of SN-based test statistic $T_n^{\diamond}$ in \cite{wang2019inference} based on $1000$ replications. 

\begin{table}[!th]
\centering
\caption{Size and power of multiple change points testing}
\label{multiresult}
\begin{tabular}{lrrrrrrrrrrrr}
\hline
$(n,p) = (50,50)$ & \multicolumn{4}{c}{$H_0$}                                                                               & \multicolumn{4}{c}{$H_1(2CP)$}                                                                          & \multicolumn{4}{c}{$H_2(3CP)$}                                                                          \\ \cline{2-13} 
$AR(0.5)$         & \multicolumn{2}{c}{SN}                             & \multicolumn{2}{c}{Boot}                           & \multicolumn{2}{c}{SN}                             & \multicolumn{2}{c}{Boot}                           & \multicolumn{2}{c}{SN}                             & \multicolumn{2}{c}{Boot}                           \\ \hline
$\alpha$          & \multicolumn{1}{c}{0.05} & \multicolumn{1}{c}{0.1} & \multicolumn{1}{c}{0.05} & \multicolumn{1}{c}{0.1} & \multicolumn{1}{c}{0.05} & \multicolumn{1}{c}{0.1} & \multicolumn{1}{c}{0.05} & \multicolumn{1}{c}{0.1} & \multicolumn{1}{c}{0.05} & \multicolumn{1}{c}{0.1} & \multicolumn{1}{c}{0.05} & \multicolumn{1}{c}{0.1} \\ \hline
A0                & 0.143                    & 0.206                   & 0.039                    & 0.098                   & 0.253                    & 0.335                   & 0.166                    & 0.321                   & 0.242                    & 0.320                   & 0.080                    & 0.194                   \\
A1                & 0.222                    & 0.318                   & 0.037                    & 0.094                   & 0.999                    & 0.999                   & 0.839                    & 0.934                   & 1.000                    & 1.000                   & 0.493                    & 0.727                   \\
A2                & 0.384                    & 0.495                   & 0.040                    & 0.097                   & 0.994                    & 0.997                   & 0.448                    & 0.645                   & 1.000                    & 1.000                   & 0.189                    & 0.366                   \\
A1+A2             & 0.242                    & 0.329                   & 0.043                    & 0.106                   & 0.997                    & 1.000                   & 0.601                    & 0.777                   & 1.000                    & 1.000                   & 0.268                    & 0.487                   \\ \hline
\end{tabular}

\begin{tabular}{lrrrrrrrrrrrr}
\hline
$(n,p) = (50,50)$ & \multicolumn{4}{c}{$H_0$}                                                                               & \multicolumn{4}{c}{$H_1(2CP)$}                                                                          & \multicolumn{4}{c}{$H_2(3CP)$}                                                                          \\ \cline{2-13} 
$AR(0.8)$         & \multicolumn{2}{c}{SN}                             & \multicolumn{2}{c}{Boot}                           & \multicolumn{2}{c}{SN}                             & \multicolumn{2}{c}{Boot}                           & \multicolumn{2}{c}{SN}                             & \multicolumn{2}{c}{Boot}                           \\ \hline
$\alpha$          & \multicolumn{1}{c}{0.05} & \multicolumn{1}{c}{0.1} & \multicolumn{1}{c}{0.05} & \multicolumn{1}{c}{0.1} & \multicolumn{1}{c}{0.05} & \multicolumn{1}{c}{0.1} & \multicolumn{1}{c}{0.05} & \multicolumn{1}{c}{0.1} & \multicolumn{1}{c}{0.05} & \multicolumn{1}{c}{0.1} & \multicolumn{1}{c}{0.05} & \multicolumn{1}{c}{0.1} \\ \hline
A0                & 0.226                    & 0.288                   & 0.052                    & 0.116                   & 0.319                    & 0.386                   & 0.125                    & 0.245                   & 0.299                    & 0.371                   & 0.091                    & 0.180                   \\
A1                & 0.307                    & 0.386                   & 0.045                    & 0.122                   & 0.972                    & 0.990                   & 0.435                    & 0.609                   & 0.994                    & 0.997                   & 0.230                    & 0.389                   \\
A2                & 0.413                    & 0.524                   & 0.066                    & 0.121                   & 0.914                    & 0.946                   & 0.211                    & 0.380                   & 0.968                    & 0.981                   & 0.142                    & 0.263                   \\
A1+A2             & 0.314                    & 0.397                   & 0.062                    & 0.125                   & 0.935                    & 0.954                   & 0.266                    & 0.454                   & 0.966                    & 0.981                   & 0.161                    & 0.295                   \\ \hline
\end{tabular}
\end{table}
According to Table 2, we can see that even for homoscedastic case (trend type A0), the SN-based test  is unable to control the size, which is presumably due to relatively small sample size $n$ and dimension $p$. The size distortion for the SN-based test when there are time varying heterocedasticity is obvious.  In comparison, the bootstrap test shows quite accurate size for both homoscedastic and  heteroscedastic cases. Notice that under the alternatives where there are 2 or 3 change points, the proposed bootstrap test still shows respectable power in most cases. Due to the size inflation of SN-based test, the interpretation of its high power needs to be done with caution. Overall, we would not recommend to use SN-based test when there is time varying heteroscedasticity and bootstrap-based test is preferred.

\subsection{Estimation}
\label{sub:estimation}

In this subsection, we examine the finite sample performance of the WBS+Bootstrap based change point estimation method described in Section 4. We followed the same setting used in \cite{wang2019inference}. Let $n=120$, $p=50$, and change point locations are $30,60$ and $90$. These change points partitioned the data into four zones. We draw i.i.d. normal samples from $N(\bm \nu_j,\bm I_p), j = 1, 2, 3, 4$ for each zone. Let $\bm \theta_j = \bm \nu_{j+1} -\bm \nu_j$ be the strength of the signals. For the dense case, we choose $\bm \theta_1 = k\times \bm 1_p$, $\bm \theta_2 = -k\times \bm 1_p$, $\bm \theta_3 = k\times \bm 1_p$ and $k = \sqrt{2.5/p}, 2\sqrt{2.5/p}$. We consider all the trend functions.

In addition to reporting the frequency for the difference between the estimated number of change points and the actual number of change points ($\hat N - N $), we also use the mean squared error (MSE) of $(\hat N - N)$ to measure the estimation accuracy for the number of change point. Similar to the comparison in \cite{wang2019inference}, we can view the change point estimation problem as a special case of classification. We treat the data between two successive change points as if they are in the same category, and evaluate the classification accuracy based on Adjusted Rand Index (ARI) \citep{rand1971objective,hubert1985comparing,wang2018high}. ARI can only take values between 0 and 1, and the larger ARI is associated with the better accuracy. When all change points are estimated perfectly, the ARI is 1. If there is no change point estimated, the corresponding ARI is 0.  The results are summarized in Table 3. 
Notice that for weaker signal $k = \sqrt{2.5/p}$, even when there is no trend (Type A0), the WBS+SN is unable to provide an accurate estimate while our WBS+Bootstrap correctly estimates the number and location of change points. When there is heteroscedasticity, our WBS+Bootstrap also substantially outperforms WBS+SN, in particular for trend types A2, A4, A1+A2 and A1+ A4. Both methods perform worse for trend type A2 which corresponds to a linear trend, while the WBS+Bootstrap still maintains reasonably good MSE and ARI. For the strong signal $k = 2\sqrt{2.5/p}$ case, WBS+SN  still cannot compete with WBS+Bootstrap when there is no trend (A0).  For the heteroscedastic cases, WBS+SN seems to perform better than the no trend case. However, this is due to the fact that the trend function yields a smaller variance, which makes the signal to noise ratio larger and easier for WBS+SN to estimate the change point locations. 
For all four trend types we considered,  the estimated number of change points $\hat N$ by WBS+Bootstrap are all correct in 200 replications (i.e., MSE = 0), which is another evidence to support the superiority of WBS+Bootstrap over WBS+SN.

\begin{table}[]
\label{tab:WBS}
\caption{WBS for change point estimation}
\begin{tabular}{llrrrrrrr}
\hline
$\sqrt{2.5/p}$                     &      & \multicolumn{5}{c}{$\hat N - N$}                                                                                         & \multicolumn{1}{c}{MSE} & \multicolumn{1}{c}{ARI} \\ \cline{3-9} 
                         &      & \multicolumn{1}{c}{-3} & \multicolumn{1}{c}{-2} & \multicolumn{1}{c}{-1} & \multicolumn{1}{c}{0} & \multicolumn{1}{c}{1} & \multicolumn{1}{c}{}    & \multicolumn{1}{c}{}    \\ \hline
\multirow{2}{*}{A0}      & SN   & 196                    & 4                      & 0                      & 0                     & 0                     & 8.900                   & 0.005                   \\
                         & Boot & 90                       &   64                     &    34                    &   12                    &    0                   &   5.504                      &   0.277                      \\ \hline
\multirow{2}{*}{A1}      & SN   & 0                      & 6                      & 161                    & 33                    & 0                     & 0.926                   & 0.723                   \\
                         & Boot & 0                      & 0                      & 0                      & 200                   & 0                     & 0.000                   & 0.985                   \\ \hline
\multirow{2}{*}{A2}      & SN   & 0                      & 88                     & 110                    & 2                     & 0                     & 2.311                   & 0.524                   \\
                         & Boot & 2                      & 12                     & 46                     & 140                   & 0                     & 0.564                   & 0.865                   \\ \hline
\multirow{2}{*}{A3}      & SN   & 0                      & 0                      & 0                      & 198                   & 2                     & 0.010                   & 0.967                   \\
                         & Boot & 0                      & 0                      & 0                      & 200                   & 0                     & 0.000                   & 0.999                   \\ \hline
\multirow{2}{*}{A4}      & SN   & 15                     & 65                     & 81                     & 28                    & 1                     & 2.389                   & 0.618                   \\
                         & Boot & 0                      & 0                      & 0                      & 200                   & 0                     & 0.000                   & 0.984                   \\ \hline
\multirow{2}{*}{A1 + A2} & SN   & 0                      & 54                     & 136                    & 10                    & 0                     & 1.761                   & 0.594                   \\
                         & Boot & 0                      & 2                      & 16                     & 182                   & 0                     & 0.121                   & 0.949                   \\ \hline
\multirow{2}{*}{A1 + A3} & SN   & 0                      & 2                      & 90                     & 107                   & 1                     & 0.496                   & 0.826                   \\
                         & Boot & 0                      & 0                      & 0                      & 199                   & 1                     & 0.005                   & 0.993                   \\ \hline
\multirow{2}{*}{A1 + A4} & SN   & 11                     & 43                     & 106                    & 40                    & 0                     & 1.888                   & 0.634                   \\
                         & Boot & 0                      & 0                      & 0                      & 200                   & 0                     & 0.000                   & 0.987                   \\ \hline
\end{tabular}
\begin{tabular}{llrrrrrrr}
\hline
$2\sqrt{2.5/p}$                       &      & \multicolumn{5}{c}{$\hat N - N$}                                                                                         & \multicolumn{1}{c}{MSE} & \multicolumn{1}{c}{ARI} \\ \cline{3-9} 
                         &      & \multicolumn{1}{c}{-3} & \multicolumn{1}{c}{-2} & \multicolumn{1}{c}{-1} & \multicolumn{1}{c}{0} & \multicolumn{1}{c}{1} & \multicolumn{1}{c}{}    & \multicolumn{1}{c}{}    \\ \hline
\multirow{2}{*}{A0}      & SN   & 31                     & 55                     & 70                     & 43                    & 1                     & 2.855                   & 0.538                   \\
                         & Boot & 0                      & 0                      & 0                      & 200                   & 0                     & 0.000                   & 0.986                   \\ \hline
\multirow{2}{*}{A1}      & SN   & 0                      & 0                      & 0                      & 198                   & 2                     & 0.010                   & 0.975                   \\
                         & Boot & 0                      & 0                      & 0                      & 200                   & 0                     & 0.000                   & 0.985                   \\ \hline
\multirow{2}{*}{A2}      & SN   & 0                      & 0                      & 2                      & 197                   & 1                     & 0.015                   & 0.964                   \\
                         & Boot & 0                      & 0                      & 0                      & 200                   & 0                     & 0.000                   & 0.998                   \\ \hline
\multirow{2}{*}{A3}      & SN   & 0                      & 0                      & 0                      & 198                   & 2                     & 0.010                   & 0.981                   \\
                         & Boot & 0                      & 0                      & 0                      & 200                   & 0                     & 0.000                   & 0.999                   \\ \hline
\multirow{2}{*}{A4}      & SN   & 0                      & 0                      & 0                      & 198                   & 2                     & 0.010                   & 0.971                   \\
                         & Boot & 0                      & 0                      & 0                      & 200                   & 0                     & 0.000                   & 0.984                   \\ \hline
\multirow{2}{*}{A1 + A2} & SN   & 0                      & 0                      & 0                    & 196                    & 4                     & 0.020                   & 0.970                   \\
                         & Boot & 0                      & 0                      & 0                      & 200                   & 0                     & 0.000                   & 0.999                   \\ \hline
\multirow{2}{*}{A1 + A3} & SN   & 0                      & 0                      & 0                      & 198                   & 2                     & 0.010                   & 0.975                   \\
                         & Boot & 0                      & 0                      & 0                      & 200                   & 0                     & 0.000                   & 0.993                   \\ \hline
\multirow{2}{*}{A1 + A4} & SN   & 0                      & 0                      & 0                      & 199                   & 1                     & 0.005                   & 0.968                   \\
                         & Boot & 0                      & 0                      & 0                      & 200                   & 0                     & 0.000                   & 0.987                   \\ \hline
\end{tabular}
\centering
\end{table}

\section{Real data application}
\label{sec:data}

In this section, we compare the performance of the proposed change point location estimation method on the micro-array bladder tumor dataset. The ACGH (Array Comparative Genomic Hybridisation) data is publicly available and it contains log intensity ratio measurements for 43 individuals at 2215 different loci on their genome. The dataset is available in R package ``ecp'' and was also studied by \cite{wang2018high} and \cite{wang2019inference}. Following the latter paper, we only considered first 200 loci and perform change point estimation using WBS+Bootstrap and compare with WBS+SN.

To examine whether there are changes in the variance of each component, we apply the test for constant variance proposed by \cite{schmidt2020asymptotic} for a univariate time series to each of the 43 subjects. For a sequence of univariate random variable $D_1,\ldots,D_n$, the test statistic is constructed as follows:
\[U(n) = \frac{1}{b_n(b_n-1)}\sum_{1\le j\ne k \le b_n}|\log \hat \sigma^2_j - \log \hat \sigma^2_k|,\]
where
\[\hat \sigma^2_j = \frac{1}{l_n}\sum_{i = (j-1)l_n+1}^{jl_n}\left( D_i - \frac{1}{l_n}\sum_{r = (j-1)l_n+1}^{jl_n} D_r\right)^2,\quad l_n = \floor{n^s}, \quad b_n = \floor{n/l_n}.\]
Under the null, $U(n)$ is asymptotically normal,
\[\sqrt{b_n}\left(\frac{\sqrt{l_n}}{\hat \kappa^*}U(n)-\frac{2}{\sqrt{\pi}}\right) \xrightarrow{D} N\left(0,\frac{4}{3}+\frac{8}{\pi}(\sqrt{3}-2)\right),\]
where $\hat \kappa^{*2}$ is the estimated long run variance
\[\hat \kappa^* = \frac{1}{\Tilde{b}_n}\sqrt{\frac{\pi}{2}}\frac{1}{\hat \sigma^2_H} \sum_{j = 1}^{\Tilde{b}_n}\left| \frac{1}{\sqrt{\Tilde{l}_n}}\sum_{i = (j-1)\Tilde{l}_n +1}^{j \Tilde{l}_n}(\Tilde{D}_i^2 - \hat \sigma_H^2)\right|, \quad \Tilde{D}_i = D_i - \frac{1}{l_n}\sum_{r = (j-1)l_n + 1}^{jl_n} D_r,  \quad  \hat \sigma_H^2 = \frac{1}{n}\sum_{i = 1}^n\Tilde{D}_i^2, \]
\[\Tilde{l}_n = \floor{n^q}, \quad \Tilde{b}_n = \floor{n/\Tilde{l}_n}.\]
We set the tuning parameters $s = 0.7$ and $q = 0.5$, following the recommendation in \cite{schmidt2020asymptotic}. An appealing feature of this test is that it allows for changes in the mean, in particular, a piecewise Lipschitz-continuous mean function. We treat the resulting 43 p-values as independent and apply the Higher Criticism test~\citep{donoho2004} to determine whether there is a variance change in any of the 43 dimensions. The resulting p-value is $0.022$, which indicates quite strong evidence against the constant variance assumption for all components. The WBS+SN  yields 6 change points $\{ 39,74,87, 134, 173,191\}$, while the WBS+Bootstrap  only reports 3 change points at $\{73, 135, 173\}$, which largely coincides with the three change-points $\{74,134,173\}$ detected by WBS+SN. The additional change-point locations obtained from WBS+SN could be spurious due to the variance instability un-accounted for in the latter procedure.


\section{Conclusion}
\label{sec:conclusion}

In this paper, we develop a bootstrap-based test for the mean changes in high-dimensional  heteroscedastic data. Existing literature on high-dimensional mean change detection  exclusively focuses on the homoscedastic case, and the applicability of existing tests is questionable when there is time-varying heterscedasticity. Building on the U-statistic approach proposed in \cite{wang2019inference}, we develop a new test statistic and a bootstrap-based approximation for single change point testing. The bootstrap consistency is justified under mild assumptions on the heteroscedasticity and componentwise dependence. Our test involves no tuning parameters and is easy to implement.  Extensions to multiple change-points testing and estimation using WBS are also presented.  Numerical comparison demonstrates the robustness of our proposed testing and estimation procedures with respect to time-varying heteroscedasticity and the degree of panel dependence. 


To conclude, we mention a few possible extensions. First, it would be interesting to extend our method to allow temporal dependence, that is, assuming $\{Z_i\}$ to be  stationary and weakly dependent instead of independent observations. Under this setting, the Gaussian multiplier bootstrap may not be adequate. The dependent wild bootstrap proposed in \cite{shao2010wdependent} may be needed to capture the serial dependence. Second,  as the numerical results suggest, the bootstrap-based test may still work when the panel dependence is strong, i.e., compound symmetric case. It would be desirable to expand our theory to cover this interesting case. Third, we did not provide any theoretical support for the consistency of WBS+Bootstrap, although the empirical performance is very encouraging. Further theoretical investigation is left for future work.


\newpage

\appendix

\section{Appendix}
\label{sec:appendix}

In the appendix, we include all the technical proofs for the theorems. Note that under $H_0$, the test statistics $T_n$ can be viewed a continuous transformation of a partial sum process
\[
S_n(a,b) = \sum_{i =\floor{na}+1}^{\floor{nb}-1} \sum_{j = \floor{na}+1}^i X_{i+1}^TX_{j}
\]
for any $0 \le a < b \le 1$ and $\floor{na}+1 \le \floor{nb}-1$.

Consider the following representation of the bootstrapped version of the partial sum process $S^*_n(a,b)$:
\begin{align*}
S_n^*(a,b) &= \sum_{i = \floor{na}+1}^{\floor{nb}-1} \sum_{j= \floor{na}+1}^i (X_{i+1} - \bar X)^T (X_j - \bar X) e_{i+1}e_j
\\
&=  \sum_{i = \floor{na}+1}^{\floor{nb}-1} \sum_{j= \floor{na}+1}^i  X_{i+1}^TX_j e_{i+1}e_j  - \bar X^T \sum_{i = \floor{na}+1}^{\floor{nb}-1} \sum_{j= \floor{na}+1}^i  X_j e_{i+1}e_j \\
& \quad - \bar X^T \sum_{i = \floor{na}+1}^{\floor{nb}-1} \sum_{j= \floor{na}+1}^i  X_{i+1} e_{i+1}e_j + \bar X^T \bar X \sum_{i = \floor{na}+1}^{\floor{nb}-1} \sum_{j= \floor{na}+1}^i e_{i}e_j\\
& = S_{n,1}^*(a,b) + S_{n,2}^*(a,b) + S_{n,3}^*(a,b) + S_{n,4}^*(a,b)
\end{align*}
The proofs are divided into three subsections: In Section 8.1, we show the bootstrap process $S_n^*(a,b)$ converges to the same limiting process of $S_n(a,b)$ by using the unconditional convergence argument proposed in \cite{bucher2019note}. The asymptotic results in Theorem 1 and  Theorem 2 follow from these arguments. In Section 8.2, we study the behavior of  $S_n(a,b)$ and $S_n^*(a,b)$ under three different kinds of alternatives. The power of the bootstrap test presented in Theorem 3 follows from these results. Finally, in Section 8.3, we show theoretical results for multiple change points testing, which is a generalization of the first two parts.


\subsection{Proof of Theorems~\ref{thm:1} and~\ref{thm:2}}

\subsubsection{Process convergence of $S_n$ and $S_n^*$ under the null.}

This section contains the crucial technical ingredient for establishing bootstrap consistency. Let
\[
S_n^{k,*}(a,b) := \sum_{i = \floor{na}+1}^{\floor{nb}-1} \sum_{j= \floor{na}+1}^i (X_{i+1} - \bar X)^T (X_j - \bar X) e_{i+1,k}e_{j,k}
\]
where $\{e_{i,k}\}_{i=1,\dots,n}, k= 1,2$ denote two independent collections of i.i.d. $N(0,1)$ random variables. The main result in this section establishes process convergence of $S_n$ and joint convergence of $(S_n,S_n^{1,*},S_n^{2,*})$ under the null. The latter result will be later combined with and the results in \cite{bucher2019note} to establish bootstrap consistency under the null.  

\begin{proposition}\label{prop:convSn}

Let Assumptions~\ref{as:1}-\ref{as:4} hold. Then
\[
\left\{ \frac{1}{n \|\Sigma\|_F} S_n(a,b)\right\}_{(a,b) \in[0,1]^2} \rightsquigarrow Q \text{ in } \ell^{\infty}([0,1]^2),
\]
where the centered Gaussian process $Q$ is defined in Theorem~\ref{thm:1}. If Assumption~\ref{as:5} also holds then
\begin{equation}\label{eq:SnJoint}
\left(\left\{ \frac{S_n(a,b)}{n \|\Sigma\|_F} \right\}_{(a,b) \in[0,1]^2},\left\{ \frac{S_n^{1,*}(a,b)}{n \|\Sigma\|_F} \right\}_{(a,b) \in[0,1]^2},\left\{ \frac{S_n^{2,*}(a,b)}{n \|\Sigma\|_F} \right\}_{(a,b) \in[0,1]^2}\right) \rightsquigarrow (Q,Q^{(1)},Q^{(2)})
\end{equation}
where $Q,Q^{(1)},Q^{(2)}$ are i.i.d. copies of $Q$ and convergence takes place in $\ell^{\infty}([0,1]^2)\times\ell^{\infty}([0,1]^2)\times\ell^{\infty}([0,1]^2)$. Moreover, the sample paths of each process are asymptotically uniformly equicontinuous in probability with respect to the Euclidean metric in $[0,1]^2$.
\end{proposition}

The proof of Proposition~\ref{prop:convSn} is long and technical and will be split over several subsections. Since the proof of the second statement contains the proof of the first statement, we will only provide that proof. A close look will reveal that all parts which are relevant to showing the first part go through without Assumption~\ref{as:5}.


\begin{proof}[Proof of Theorem~\ref{thm:1}]
Begin by observing the representation
\begin{align*}
\tilde G_n(k) = ~& \frac{2(n-k)(n-k-1)}{n^3} \widetilde S_n(1,k) + \frac{2 k(k-1)}{n^3}\widetilde S_n(k+1,n)
\\
& - \frac{2k(n-k)}{n^3} (\widetilde S_n(1,n) - \widetilde S_n(1,k) - \widetilde S_n(k+1,n))
\end{align*}
where
\[
\widetilde S_n(k,m) := \sum_{i = k}^{m}\sum_{j = k}^{i} X_{i+1}^TX_j.
\]
Proposition~\ref{prop:convSn} and uniform asymptotic equi-continuity of the sample path of $S_n$ in probability together with some simple calculations yields,
\[
\frac{1}{\|\Sigma\|_F}\Tilde{G}_n(\floor{nr})\rightsquigarrow G(r) := 2(1-r)Q(0,r) + 2rQ(r, 1) - 2r(1-r) Q(0, 1).
\]
Since the sample paths of $Q$ are uniformly continuous with respect to the Euclidean metric on $[0,1]^2$, a simple calculation shows that the sample paths of $G(r;0,1)$ are uniformly continuous with respect to the Euclidean metric on $[0,1]$. Consider the maps
\[
\Phi_n(f) := \max_{k = 2,\ldots,n-3} f(k/n),
\]
defined for bounded functions $f:[0,1] \to \mathbb{R}$. With this definition, we have $T_n = \Phi_n(\Tilde{G}_n)$.
Consider the map 
\[
\Phi(f) = \sup_{r\in[0,1]} f(r)
\]
defined for bounded functions $f: [0,1] \to \R$. It is straightforward to see that, for any sequence of bounded functions $f_n$ with $\|f_n - f\|_\infty = o(1)$ for a continuous function $f$, we have $\Phi_n(f_n) \to \Phi(f)$. Applying the extended continuous mapping theorem (see Theorem 1.11.1 in \cite{van1996weak}), this implies 
\[
\frac{1}{\|\Sigma\|_F}T_n = \Phi_n\Big(\frac{1}{\|\Sigma\|_F} \Tilde{G}_n\Big) \rightsquigarrow \Phi(G) = \sup_{r\in [0,1]}G(r) =T,
\] 
which completes the proof.
\end{proof}

\noindent
\textit{Proof of Theorem \ref{thm:2}} 
Similar arguments as given in the proof of Theorem~\ref{thm:1} but utilizing equation~\eqref{eq:SnJoint} instead of the first part of that proposition show that
\begin{equation}\label{eq:tnbjoint}
\Big(\frac{T_n}{\|\Sigma\|_F},\frac{T_{n,1}^*}{\|\Sigma\|_F}, \frac{T_{n,2}^*}{\|\Sigma\|_F}\Big) \weak (T,T',T'')
\end{equation}
where $T', T''$ are i.i.d. copies of $T$ and $T_{n,1}^*,T_{n,2}^*$ are two copies of the bootstrap statistic each with independent sets of multipliers $e_i$.

Next observe that by Corollary 1.3 and Remark 4.1 in \cite{gaenssler2007continuity} the function
\[
t \mapsto P\Big( \sup_{r \in [0,1]} |G(r)| \leq t \Big)
\] 
is continuous on $\R$ and strictly increasing on $\R^+$. This implies that the function 
\[
H(t) := P\Big( \sup_{r \in [0,1]} G(r) \leq t \Big)
\]
satisfies $H(\eps) > 0$ for all $\eps >0$. Indeed,
\[
P\Big( \sup_{r \in [0,1]} G(r) \leq \eps \Big) \geq P\Big( \sup_{r \in [0,1]} |G(r)| \leq \eps \Big) > 0
\]
since the latter is strictly increasing on $\R^+$. Thus the left support point of the cdf $H$ must be in $t \leq 0$. Hence by Theorem 1 in \cite{tsirel1976density} the function $H$ is continuous on $(0,\infty)$. Clearly $H(0) \leq 1 - P(G(1/2) > 0) = 1/2$.  

The proof is completed by observing that the conclusion of Lemma 4.2 in \cite{bucher2019note} remains true if continuity of the cdf $F$ in there (corresponding to $H$ in our case) is replaced by continuity on $(0,\infty)$ and the additional assumption $G^{-1}(1-\alpha) \in (0,\infty)$ is made (note that $1-\alpha$ in our notation corresponds to $\alpha$ in \cite{bucher2019note}). The condition $G^{-1}(1-\alpha) \in (0,\infty)$ is guaranteed by the assumption $\alpha < 1/2$. Next observe that~\eqref{eq:tnbjoint} verifies Condition (a) in Lemma 2.2 in the latter paper. Condition 4.1 in \cite{bucher2019note} is satisfied as well and relaxing continuity of the cdf of the limit was described above. This completes the proof. \hfill $\Box$

\subsection{Proof of Proposition~\ref{prop:convSn}}\label{sec:pfconvss*}

We begin by providing an overview of the proof: first, we show that the process $S_{n}^{k,*}$ admits the representation
\begin{equation}\label{eq:SnStarrepr}
S_n^{k,*}(a,b) = \sum_{i = \floor{na}+1}^{\floor{nb}-1} \sum_{j= \floor{na}+1}^i X_{i+1}^TX_j e_{i+1}e_j + o_P(n\|\Sigma\|_F) =: S_{n,1}^{k,*}(a,b) + o_P(n\|\Sigma\|_F)
\end{equation}
uniformly in $a,b \in [0,1]$, see section~\ref{sec:pfSn*repr}. Thus it suffices to establish~\eqref{eq:SnJoint} with $S_{n,1}^{k,*}$ instead of $S_{n}^{k,*}$. Then, in section~\ref{sec:asyeq}, we show that under Assumptions~\ref{as:1}--\ref{as:4},
\begin{align}  \label{eq:asyeqSn}
\lim_{\delta \downarrow 0} \limsup_{n\rightarrow \infty} P\Big( \sup_{\|u-v\|_2 \le \delta} \Big|\frac{1}{n||\Sigma||_F}S_n(u)- \frac{1}{n\|\Sigma\|_F}S_n(v)\Big| > x\Big) = 0.
\\ \label{eq:asyeqSnStar}
\lim_{\delta \downarrow 0} \limsup_{n\rightarrow \infty} P\Big( \sup_{\|u-v\|_2 \le \delta} \Big|\frac{1}{n||\Sigma||_F}S_{n,1}^{1,*}(u)- \frac{1}{n\|\Sigma\|_F}S_{n,1}^{1,*}(v)\Big| > x\Big) = 0.
\end{align}
This implies that each of the processes $S_n/n\|\Sigma\|_F$, $S_{n,1}^{1,*}/n\|\Sigma\|_F$, $S_{n,1}^{2,*}/n\|\Sigma\|_F$ is tight. Finally, we show joint finite-dimensional convergence of the processes $S_n/n\|\Sigma\|_F, S_{n,1}^{1,*}/n\|\Sigma\|_F, S_{n,1}^{2,*}/n\|\Sigma\|_F$ to the joint limit in~\eqref{eq:SnJoint}, again under Assumptions~\ref{as:1}--\ref{as:4} (see section \ref{sec:sn*fidi}). Combined, the results above imply the statement in~\eqref{eq:SnJoint}. Note in particular that process convergence of $S_n/(n\|\Sigma\|_F)$ follows under just Assumptions~\ref{as:1}--\ref{as:4} without utilizing Assumption~\ref{as:5}.

Before proceeding, we state a useful technical Lemma that we will utilize in several places throughout the proof.

\begin{lemma}\label{lem:boundmom}
Under Assumptions~\ref{as:2} and \ref{as:3} we have for a constant $\tilde C$ independent of $n,p$ and the distribution of $Z$ we have for $s=4,6$
\[
\max_{j_1,\dots,j_s,k_1,\dots,k_s = 1, k_i \neq j_i} \Big| E[X_{k_1}^TX_{j_1}\cdots X_{k_s}^TX_{j_s}] \Big| \leq \tilde CB^{2s} \|\Sigma\|_F^s
\]
\end{lemma} 
\textit{Proof}: 
Observe that by the generalized version of H\"older's inequality 
\begin{align*}
\Big| E[X_{k_1}^TX_{j_1}\cdots X_{k_s}^TX_{j_s}] \Big|
\leq E[|X_{k_1}^TX_{j_1}|^s]^{1/s}\cdots E[|X_{k_s}^TX_{j_s}|^s]^{1/s}
\leq \max_{k \neq j} E[|X_{k}^TX_{j}|^s]
\end{align*}
Now let $\mathcal{P}_s$ denote the set of disjoint partitions $\pi$ of the set $1,\dots,s$ such that 
With this notation we obtain for $k \neq j$ 
\begin{align*}
E[|X_{k}^TX_{j}|^s]& = \Big| \sum_{l_1,\dots,l_s=1}^p H_{l_1}(k/n)H_{l_1}(j/n)\cdots H_{l_s}(k/n)H_{l_s}(j/n) E[ Z_{k,l_1}Z_{j,l_1}\cdots Z_{k,l_s}Z_{j,l_s} ] \Big|
\\
& = \Big| \sum_{l_1,\dots,l_s=1}^p H_{l_1}(k/n)H_{l_1}(j/n)\cdots H_{l_s}(k/n)H_{l_s}(j/n) \Big(E[ Z_{k,l_1} \cdots Z_{k,l_s}]\Big)^2 \Big|
\\
& \leq B^{2s}\sum_{l_1,\dots,l_s=1}^p  \Big(E[ Z_{1,l_1}\cdots Z_{1,l_s}]\Big)^2
\\
& =  B^{2s} \sum_{l_1,\dots,l_s=1}^p \Big(\sum_{\pi \in\mathcal{P}_s} \prod_{B \in \pi} cum(Z_{1,l_j}: j \in B) \Big)^2
\\
& \leq B^{2s} |\mathcal{P}_s| \sum_{l_1,\dots,l_s=1}^p \sum_{\pi \in\mathcal{P}_s} \prod_{B \in \pi} cum(Z_{1,l_j}: j \in B)^2 
\\
& = B^{2s} |\mathcal{P}_s| \sum_{\pi \in\mathcal{P}_s} \prod_{B \in \pi} \Big\{ \sum_{l_k: k \in B} cum(Z_{1,l_j}: j \in B)^2 \Big\}
\\
& \leq  B^{2s} |\mathcal{P}_s| \sum_{\pi \in\mathcal{P}_s} \prod_{B \in \pi} \|\Sigma\|_F^{|B|} 
\\
& = B^{2s} |\mathcal{P}_s| \|\Sigma\|_F^s 
\end{align*}
where the second equality uses stationarity and independence across $t$ of $\{Z_t\}$ and the last inequality follows by Assumption~\ref{as:2}. Setting $\tilde C = |\mathcal{P}_s|$ completes the proof. 

\hfill $\Box$


\subsubsection{Proof of~\eqref{eq:asyeqSn} and~\eqref{eq:asyeqSnStar}} \label{sec:asyeq}

Both proofs follow the same principle. Observe that the processes $S_n, S_{n,1}^*$ are piecewise constant on their index set and their values are entirely determined by their values on the grid $\{(i/n,j/n): i,j=0,\dots,n\}$. Now following the arguments in section 8.8.1 in  \cite{wang2019inference} it is clear that \eqref{eq:asyeqSn} and~\eqref{eq:asyeqSnStar} follow if we prove that there exists a constant $C$ which is independent of $n$ such that
\begin{align*}
\sup_{u,v \in [0,1]^2}E\Big[\frac{|S_n(u) - S_n(v)|^6}{n^6\|\Sigma\|_F^6}\Big] \leq C (\|u-v\|_2^3 + n^{-3}),
\\
\sup_{u,v \in [0,1]^2}E\Big[\frac{|S_{n,1}^{*}(u) - S_{n,1}^{*}(v)|^6}{n^6\|\Sigma\|_F^6}\Big] \leq C (\|u-v\|_2^3 + n^{-3}).
\end{align*} 
Next, a close look at the proof of (8.18) in \cite{wang2019inference} shows that it suffices to show that, for a possibly different constant $C$,
\begin{align*}
&\max_{j_1,\dots,j_6,k_1,\dots,k_6 = 1, k_i \neq j_i} \Big| E[X_{k_1}^TX_{j_1}\cdots X_{k_6}^TX_{j_6}] \Big| \leq  C \|\Sigma\|_F^6,
\\
&\max_{j_1,\dots,j_6,k_1,\dots,k_6 = 1, k_i \neq j_i} \Big| E[X_{k_1}^TX_{j_1}\cdots X_{k_6}^TX_{j_6} e_{k_1}e_{j_1}\cdots e_{k_6} e_{j6}] \Big| \leq  C \|\Sigma\|_F^6.
\end{align*}
The first bound is a direct consequence of Lemma~\ref{lem:boundmom}. For the second bound, note that by independence of $e_i$ and $X_i$
\[
\Big| E[X_{k_1}^TX_{j_1}\cdots X_{k_6}^TX_{j_6} e_{k_1}e_{j_1}\cdots e_{k_6} e_{j6}] \Big| = \Big| E[X_{k_1}^TX_{j_1}\cdots X_{k_6}^TX_{j_6}] E[ e_{k_1}e_{j_1}\cdots e_{k_6} e_{j6}] \Big|
\]
and the claim follows from Lemma~\ref{lem:boundmom} since the $e_i$ are standard normal and have finite moments of all orders. Note that the proofs in this section did not make use of Assumption~\ref{as:5} and all arguments hold under Assumptions~\ref{as:1}--\ref{as:4}. \hfill $\Box$ 


\subsubsection{Proof of~\eqref{eq:SnStarrepr}}\label{sec:pfSn*repr}

Throughout this section we will drop the index $k$ in $S_n^{k,*}$ for notational convenience. Consider the decomposition 
\begin{align*}
S_n^*(a,b) &= \sum_{i = \floor{na}+1}^{\floor{nb}-1} \sum_{j= \floor{na}+1}^i (X_{i+1} - \bar X)^T (X_j - \bar X) e_{i+1}e_j
\\
&=  \sum_{i = \floor{na}+1}^{\floor{nb}-1} \sum_{j= \floor{na}+1}^i  X_{i+1}^TX_j e_{i+1}e_j  - \bar X^T \sum_{i = \floor{na}+1}^{\floor{nb}-1} \sum_{j= \floor{na}+1}^i  X_j e_{i+1}e_j \\
& \quad - \bar X^T \sum_{i = \floor{na}+1}^{\floor{nb}-1} \sum_{j= \floor{na}+1}^i  X_{i+1} e_{i+1}e_j + \bar X^T \bar X \sum_{i = \floor{na}+1}^{\floor{nb}-1} \sum_{j= \floor{na}+1}^i e_{i}e_j\\
& = S_{n,1}^*(a,b) - S_{n,2}^*(a,b) - S_{n,3}^*(a,b) + S_{n,4}^*(a,b).
\end{align*}
Observe that $0 \leq \bar X^T \bar X$ and that
\[
E[\bar X^T \bar X] = \frac{1}{n^2} \sum_{i,j} E[X_i^TX_j] \leq \frac{B^2}{n} tr(\Sigma)
\]
since we are under the null and the $X_i$ are centered. Thus
\begin{equation}\label{eq:xbartxbar}
\bar X^T \bar X = O_P(tr(\Sigma)/n)
\end{equation}
Moreover
\[
\sup_{a,b \in [0,1]} \Big|\frac{1}{n} \sum_{i = \floor{na}+1}^{\floor{nb}-1} \sum_{j= \floor{na}+1}^i e_{i}e_j \Big| = O_P(1)
\]  
since the above term is simply the process $S_n$ with $e_i$ instead of $X_i$ and $S_n$ converges weakly under Assumptions~\ref{as:1}--\ref{as:4} as argued in the beginning of section~\ref{sec:pfconvss*}. Hence
\[
\sup_{a,b\in [0,1]} \Big|S_{n,4}^*(a,b) \Big| = O_P(tr(\Sigma)) = o_P\Big(n\|\Sigma\|_F\Big)
\]
by Assumption~\ref{as:5}. Next observe the decomposition
\begin{align*}
S_{n,2}^*(a,b) + S_{n,3}^*(a,b)
& = \bar{X}^T \sum_{i,j = \floor{na}+1}^{\floor{nb}-1} X_j e_ie_j - \bar{X}^T \sum_{i = \floor{na}+1}^{\floor{nb}-1} X_i e_i^2
\\
& = \bar{X}^T\Big( \sum_{i = \floor{na}+1}^{\floor{nb}-1} X_ie_i \Big)  \sum_{j = \floor{na}+1}^{\floor{nb}-1} e_j - \bar{X}^T \sum_{i = \floor{na}+1}^{\floor{nb}-1} X_i e_i^2.
\end{align*}
We first deal with the second term. Observe that by the Cauchy-Schwarz inequality
\begin{align*}
\sup_{a,b \in [0,1]} \Big| \bar{X}^T \sum_{i = \floor{na}+1}^{\floor{nb}-1} X_i e_i^2 \Big|
&\leq
|\bar X^T \bar X |^{1/2} \sup_{a,b \in [0,1]}\Big| \sum_{i,j = \floor{na}+1}^{\floor{nb}-1} X_i^TX_j e_i^2e_j^2 \Big|^{1/2}
\\
& =|\bar X^T \bar X |^{1/2} \sup_{a,b \in [0,1]}\Big| \sum_{i,j = \floor{na}+1, i\neq j}^{\floor{nb}-1} X_i^TX_j e_i^2e_j^2  - \sum_{i = \floor{na}+1}^{\floor{nb}-1} X_i^TX_i e_i^4 \Big|^{1/2}
\\
&\leq |\bar X^T \bar X |^{1/2} \Big(\sup_{a,b \in [0,1]}\Big| \sum_{i,j = \floor{na}+1, i\neq j}^{\floor{nb}-1} X_i^TX_j e_i^2e_j^2\Big| + \sup_{a,b \in [0,1]} \Big|\sum_{i = \floor{na}+1}^{\floor{nb}-1} X_i^TX_i e_i^4 \Big| \Big)^{1/2}
\end{align*}
Now we have
\[
E\Big| \sum_{i = \floor{na}+1}^{\floor{nb}-1} X_i^TX_i e_i^4\Big| \leq n E| X_i^TX_i e_i^4| = O(tr(\Sigma)n).
\]
Moreover 
\[
\sup_{a,b \in [0,1]}\Big| \sum_{i,j = \floor{na}+1, i\neq j}^{\floor{nb}-1} X_i^TX_j e_i^2e_j^2\Big| = O_P(n \|\Sigma\|_F)
\]
since this is simply the process $S_n$ with the new random vectors $\tilde X_i = X_i e_i^2$. It is straightforward to check that $\tilde X_i$ satisfy Assumption~\ref{as:1}--\ref{as:4}, and thus convergence of the process follows (recall that in the beginning of Section~\ref{sec:pfconvss*} we argued that Assumptions~\ref{as:1}--\ref{as:4} suffice for process convergence of $S_n$). Combining all results so far we find that
\[
\sup_{a,b \in [0,1]} \Big| \bar{X}^T \sum_{i = \floor{na}+1}^{\floor{nb}-1} X_i e_i^2 \Big| 
= O_P(tr(\Sigma) + tr(\Sigma)^{1/2}\|\Sigma\|_F^{1/2}) = o_P(n \|\Sigma\|_F)
\]
by the assumption $tr(\Sigma) = o(n\|\Sigma\|_F)$.

Next observe that
\[
\sup_{a,b \in [0,1]} \Big| \bar{X}^T\Big( \sum_{i = \floor{na}+1}^{\floor{nb}-1} X_ie_i \Big)  \sum_{j = \floor{na}+1}^{\floor{nb}-1} e_j \Big| \leq 4 \max_{k=1,\dots,n} \Big|\bar{X}^T \sum_{i = 1}^{k} X_ie_i \Big| \times \max_{j=1,\dots,n} \Big| \sum_{i = 1}^{j} e_i \Big|
\]
By the classical Donsker theorem for partial sum processes
\[
\max_{k=1,..,n} \Big| \sum_{j = 1}^{k} e_j \Big| = O_P(n^{1/2}).
\] 
Next consider the decomposition
\[
\max_{k=1,\dots,n} \Big|\bar{X}^T \sum_{i = 1}^{k} X_ie_i \Big| \leq \max_{k=1,\dots,n}  \Big|\frac{1}{n}\sum_{i = 1}^{k} \sum_{j=1, j\neq i}^n X_j^TX_ie_i \Big| +  \max_{k=1,\dots,n}  \Big|\frac{1}{n}\sum_{i = 1}^{k} X_i^TX_ie_i \Big|.
\]
By Kolmogorov's maximal inequality
\begin{align*}
\max_{k=1,\dots,n}  \Big|\frac{1}{n}\sum_{i = 1}^{k} X_i^TX_ie_i \Big| &= O_P\Big(n^{-1/2} Var(X_i^TX_ie_i)^{1/2} \Big)
\\
&= O_P\Big(n^{-1/2} \max_i E[(X_i^TX_i)^2]^{1/2} \Big) = o_P( n^{1/2}\|\Sigma\|_F)
\end{align*}
where the last line follows since 
\begin{align*}
E[(X_j^TX_j)^2] &= \sum_{s,t =1}^pE[ H_s^2(j/n)H_t^2(j/n)Z_{j,s}^2 Z_{j,t}^2] \leq B^4 \sum_{s,t =1}^pE[Z_{j,s}^2 Z_{j,t}^2] 
\\
&= B^4\sum_{s,t =1}^p\Big( \Sigma_{s,s}\Sigma_{t,t} + \Sigma_{s,t}^2 + cum(Z_{j,s} Z_{j,s} Z_{j,t} Z_{j,t})\Big)
\\
& = B^4\tr(\Sigma)^2 + B^4\|\Sigma\|_F^2 + B^4\sum_{s,t =1}^p cum(Z_{j,s} Z_{j,s} Z_{j,t} Z_{j,t})
 = o(n^2\|\Sigma\|_F^2)
\end{align*}
by Assumption~\ref{as:5}. Hence it remains to show that 
\begin{equation}\label{eq:help1}
\max_{k=1,\dots,n}  \Big|\frac{1}{n}\sum_{i = 1}^{k} \sum_{j=1, j\neq i}^n X_j^TX_ie_i \Big| = o_P(n^{1/2}\|\Sigma\|_F).
\end{equation}
To this end observe that for $1 \leq \ell < k \leq n$
\begin{align*}
E\Big[\Big( \sum_{i = \ell}^{k} \sum_{j=1, j\neq i}^n X_j^TX_ie_i\Big)^4\Big] 
& = \sum_{j_1,\dots,j_4 = \ell}^k \sum_{k_1,\dots,k_4 = 1, k_i \neq j_i}^n E[X_{k_1}^TX_{j_1}\cdots X_{k_4}^TX_{j_4}]E[e_{j_1}e_{j_2}e_{j_3}e_{j_4}]
\\
& \leq C_1 n^2(k-\ell)^2 \max_{j_1,\dots,j_4,k_1,\dots,k_4, k_i \neq j_i} \Big| E[X_{k_1}^TX_{j_1}\cdots X_{k_4}^TX_{j_4}] \Big| 
\\
& \leq B^8 C_2 C n^2(k-\ell)^2 \|\Sigma\|_F^4
\end{align*}
where $C_1, C_2$ are constants that are independent of $n,k,\ell$ and the distribution of $X_i$ and $C$ is the constant from Assumption~\ref{as:2}. Here the last line uses Lemma~\ref{lem:boundmom}. The second-to-last line follows since the $e_i$ are centered and independent across $i$, and so are the $X_i$. Thus we can have at most two different values for the $j_i$. Further, each $k_i$ has to be equal to either at least one $j_i$ or at least one other $k_i$. This gives at most $K (k-\ell)^2n^2$ different choices for a universal constant $K$. 

To conclude the proof define the process
\[
G_n(t) := \frac{1}{n \|\Sigma\|_F} \sum_{i = 1}^{nt} \sum_{j=1, j\neq i}^n X_j^TX_ie_i 
\] 
with index set $T_n = \{i/n: i=0,\dots,n\}$ where $G_n(0) \equiv 0$. The computation above implies that for $s,t \in T_n, s\neq t$ we have (note that $s,t \in T_n, s \neq t$ implies $|s-t| \geq 1/n$)
\[
E\Big[|G_n(s) - G_n(t)|^4\Big] \leq C_3 |s-t|^2
\] 
for a universal constant $C_3$ where we used the fact that $s\neq t, s,t\in T_n$ implies $|s-t| \geq 1/n$. Applying Corollary 2.2.5 from \cite{van1996weak} with $T, \Psi, d, X$ in the latter result defined as follows: $T = T_n$, $\Psi(x) = x^4$, $d(s,t) = |t-s|^{1/2}$, $X_t = G_n(t)$ we find
\[
\Big\| \sup_{s,t \in T_n} |G_n(s) - G_n(t)| \Big\|_{4} \leq K \int_0^1 (C_4 \epsilon^{-2})^{1/4} d\epsilon < \infty,
\]
here $C_4,K$ are constants that depend on $\psi, C_3$ only and are thus independent of $n$. An application of the Markov inequality yields
\[
\sup_{s,t \in T_n} |G_n(s) - G_n(t)| = O_P(1).
\]
Finally, observe that
\[
\max_{k=1,\dots,n}  \Big|\frac{1}{n}\sum_{i = 1}^{k} \sum_{j=1, j\neq i}^n X_j^TX_ie_i \Big| =  \|\Sigma\|_F\sup_{t \in T_n} \Big| G_n(t)\Big| \leq \|\Sigma\|_F \sup_{s,t \in T_n} \Big| G_n(t) - G_n(s)\Big|= O_P(\|\Sigma\|_F)
\]
since by definition $G_n(0) = 0$. This completes the proof of~\eqref{eq:help1} and thus of~\eqref{eq:SnStarrepr} \hfill $\Box$


\subsubsection{Finite dimensional convergence result} \label{sec:sn*fidi}

\begin{proposition}	Under the Assumption~\ref{as:1}-\ref{as:5}, for any $(a_{u,k},b_{u,k}) \in (0,1)^2$, $a_{u,k}<b_{u,k}$ and contrasts $\alpha_{u,k}\in \mathbb{R}$, where $u = 1,2,3$, $k= 1, \ldots,K$, it holds that
\[
S_n := \frac{1}{n\|\Sigma\|_F}\left(\sum_{k = 1}^{K}\alpha_{1,k} S_{n}(a_{1,k}, b_{1,k})+ \sum_{k = 1}^{K}\alpha_{2,k} S_{n,1}^*(a_{2,k}, b_{2,k}) + \sum_{k = 1}^{K}\alpha_{3,k}S'^*_{n,1}(a_{3,k}, b_{3,k}) \right)
\xrightarrow{d} 
N(0,\Tilde{\sigma}^2) 
\]
where 
\[
\Tilde{\sigma}^2 = \sum_{u = 1}^3 \sum_{k = 1}^{K}\sum_{k' = 1}^{K} \alpha_{u,k}\alpha_{u,k'}V(a_{u,k}\vee a_{u,k'}, b_{u,k}\wedge b_{u,k'})
\]
\end{proposition} 

\begin{proof}
	
	Consider the following decomposition:
	\[S_n  = \sum_{i = \floor{na_{\min}}+1}^{\floor{nb_{\max}}-1} \hat \xi_{n,i+1}, \]
	where $a_{\min} = \min_{u, k}a_{u,k}$ and $b_{\max} = \max_{u,k}b_{u,k}$,
	\[\hat  \xi_{n,i+1} = \sum_{u = 1}^3\sum_{k =1}^{K}\bm 1\{\floor{a_{u,k}n} + 1 \le i\le \floor{b_{u,k}n} -1 \}\alpha_{u,k}{\xi}^u_{a_{u,k},i+1},\]
	and 
	\[{\xi}^1_{a_{u,k},i+1} =\frac{1}{n\|\Sigma\|_F} \sum_{j = \floor{a_{u,k}n}+1}^iX_{i+1}^T X_j\]
	\[{\xi}^2_{a_{u,k},i+1} = \frac{1}{n\|\Sigma\|_F} \sum_{j = \floor{a_{u,k}n}+1}^iX_{i+1}^T X_j e_{i+1}e_j\]
	\[{\xi}^3_{a_{u,k},i+1} = \frac{1}{n\|\Sigma\|_F} \sum_{j = \floor{a_{u,k}n}+1}^iX_{i+1}^T X_je'_{i+1}e'_j\]
	
	Let $\mathcal{F}_i = \sigma(X_1\ldots X_i,e_1\ldots e_i,e'_1\ldots e'_i)$ be a filtration, it is easy to check that $\sum_{i =\floor{a_{\min}n} }^j\hat{\xi}^u_{a_{u,k},i+1}$ is still a martingale. To get the convergence result, we need to check the following conditions:
	\begin{enumerate}
		\item $\forall \epsilon >0, \sum_{i =\floor{a_{\min}n}+1}^{\floor{b_{\max}n}-1} E[(\hat{\xi}_{n,i+1})^2 I(\hat{\xi}_{n,i+1} > \epsilon)|\mathcal{F}_i] \xrightarrow{p} 0$,
		\item $V_n =\sum_{i =\floor{a_{\min}n}+1}^{\floor{b_{\max}n}-1} E[(\hat{\xi}_{n,i+1})^2|\mathcal{F}_i] \xrightarrow{p} \sum_{u = 1}^3 \sum_{k = 1}^{k}\sum_{k' = 1}^{k} \alpha_{u,k} \alpha_{u, k'}V(a_{u,k}\vee a_{u,k'}, b_{u,k}\wedge b_{u,k'})$,
	\end{enumerate}
For Condition 1, it suffices to check that for any fixed interval $(a,b)$ and $u \in \{1,2,3\}$,  
\[
\sum_{i =\floor{na}+1}^{\floor{nb}-1} E\left [({\xi}^u_{a,i+1})^4\right ] \to 0.
\]
For the case $u = 1$ observe that by independence of the $X_i$ and since $X_i$ are centered
\begin{align*}
E\left [({\xi}^1_{a,i+1})^4\right ] & \frac{1}{n^4\|\Sigma\|_F^4}\sum_{i =\floor{na}+1 }^{\floor{nb}-1} E\Big[\Big(X_{i+1}^T\sum_{j =  \floor{na}+1}^i X_j\Big)^4\Big] 
\\
=& \frac{1}{n^4\|\Sigma\|_F^4}\sum_{j =  \floor{na}+1}^iE[(X_{i+1}^T X_{j})^4] + \frac{1}{n^4\|\Sigma\|_F^4}\sum_{j_1,j_2 =  \floor{na}+1}^iE[(X_{i+1}^T X_{j_1})^2 (X_{i+1}^T X_{j_2})^2] 
\\
\leq & \frac{1}{n^4\|\Sigma\|_F^4}\sum_{j =  \floor{na}+1}^iE[(X_{i+1}^T X_{j})^4] + \frac{1}{n^4\|\Sigma\|_F^4}\sum_{j_1,j_2 =  \floor{na}+1}^iE[(X_{i+1}^T X_{j_1})^4]^{1/2}E[ (X_{i+1}^T X_{j_2})^4]^{1/2}
\\
\leq & \frac{1}{n^4\|\Sigma\|_F^4}\Big(n\tilde C B^{8} \|\Sigma\|_F^4 + n^2 \tilde C B^{8} \|\Sigma\|_F^4 \Big) = O(n^{-2})
\end{align*}
where we applied Lemma~\ref{lem:boundmom} for the last step. Thus
\[
\sum_{i =\floor{na}+1}^{\floor{nb}-1} E\left [({\xi}^1_{a,i+1})^4\right ] = O(1/n) = o(1).
\]
Similarly we obtain for $u=2$
\begin{align*}
E\left [({\xi}^2_{a,i+1})^4\right ] & \frac{1}{n^4\|\Sigma\|_F^4}\sum_{i =\floor{na}+1 }^{\floor{nb}-1} E\Big[\Big(e_{i+1}X_{i+1}^T\sum_{j =  \floor{na}+1}^i e_jX_j\Big)^4\Big] 
\\
=& \frac{1}{n^4\|\Sigma\|_F^4}\sum_{j =  \floor{na}+1}^iE[(X_{i+1}^T X_{j})^4]E[e_1^4]^2 
+ \frac{1}{n^4\|\Sigma\|_F^4}\sum_{j_1,j_2 =  \floor{na}+1}^iE[(X_{i+1}^T X_{j_1})^2 (X_{i+1}^T X_{j_2})^2]E[e_1^4]
\\
\leq & \frac{9}{n^4\|\Sigma\|_F^4}\sum_{j =  \floor{na}+1}^iE[(X_{i+1}^T X_{j})^4] + \frac{3}{n^4\|\Sigma\|_F^4}\sum_{j_1,j_2 =  \floor{na}+1}^iE[(X_{i+1}^T X_{j_1})^4]^{1/2}E[ (X_{i+1}^T X_{j_2})^4]^{1/2}
\\
\leq & \frac{1}{n^4\|\Sigma\|_F^4}\Big(n9\tilde C B^{8} \|\Sigma\|_F^4 + 3n^2 \tilde C B^{8} \|\Sigma\|_F^4 \Big) = O(n^{-2})
\end{align*}
and thus 
\[
\sum_{i =\floor{na}+1}^{\floor{nb}-1} E\left [({\xi}^2_{a,i+1})^4\right ] = O(1/n) = o(1).
\]
The case $u=3$ is treated by exactly the same arguments and the proof of part 1 is complete. 

\bigskip
	
For Condition 2, observe that the bootstrap multipliers are independent of $X_i$'s, which implies
for $u \ne v$, 
\[
E[\xi^u_{a_1,i+1}\xi^v_{a_2,i+1} \mid \mathcal{F}_i] = 0.
\]
Therefore, we have the following simplification  
\begin{align*}
\sum_{i = \floor{na_{\min}} + 1}^{\floor{nb_{\max}} - 1} E[\hat{\xi}^2_{n,i+1}|\mathcal{F}_i] = \sum_{u = 1}^3 \sum_{k = 1}^{K} \sum_{k' = 1}^{K} \Big\{ \alpha_{u,k} \alpha_{u,k'} \sum_{ i = \floor{(a_{u,k} \vee a_{u,k'}) n} +1}^{\floor{(b_{u,k} \wedge b_{u,k'}) n} -1} E[\xi^u_{a_{u,k},i+1}\xi^u_{a_{u,k'},i+1}\mid \mathcal{F}_i]\Big \}.
\end{align*}
To complete the proof, it remains to show for $u \in \{1,2,3\}$,
\[
\sum_{ i = \floor{(a_{u,k} \vee a_{u,k'}) n} +1}^{\floor{(b_{u,k} \wedge b_{u,k'}) n} -1} E[\xi^u_{a_{u,k},i+1}\xi^u_{a_{u,k},i+1}\mid \mathcal{F}_i] \xrightarrow{p} V(a_{u,k}\vee a_{u,k'}, b_{u,k}\wedge b_{u,k'}).
\]
Given this structure, it suffices to show for $a' \le a \le b \le b \le 1$, 
\[
\sum_{i = \floor{an} + 1}^{\floor{bn} - 1} E[\xi^u_{a',i+1}\xi^u_{a,i+1}\mid \mathcal{F}_i] \xrightarrow{p} V(a,b).
\]
Define
\begin{align*}
M_1(a,b) &:= \sum_{i = \floor{an} + 1}^{\floor{bn} - 1} E[\xi^1_{a',i+1}\xi^1_{a,i+1}\mid \mathcal{F}_i] = \frac{1}{n^2\|\Sigma\|_F^2} \sum_{i=\floor{na}+1}^{\floor{nb}-1} E[(X_{i+1}^T\sum_{j =  \floor{na}+1}^i X_j)^2 |\mathcal{F}_i],
\\
M_2(a,b) &:= \frac{1}{n^2\|\Sigma\|_F^2} \sum_{i=\floor{na}+1}^{\floor{nb}-1} E[(X_{i+1}^Te_{i+1}\sum_{j =  \floor{na}+1}^i X_je_j)^2 |\mathcal{F}_i],
\\
M_3(a,b) &:= \frac{1}{n^2\|\Sigma\|_F^2} \sum_{i=\floor{na}+1}^{\floor{nb}-1} E[(X_{i+1}^Te'_{i+1}\sum_{j =  \floor{na}+1}^i X_je'_j)^2 |\mathcal{F}_i].
\end{align*}
Since $M_1 M_2, M_3$ have a very similar structure we will only prove that
\[
M_2(a,b) \xrightarrow{p} V(a,b), 
\]
the other two cases follow similarly. In what follows write $M_2$ for $M_2(a,b)$. Consider the following decomposition,
\begin{align*}
M_2 & = \frac{1}{n^2||\Sigma||_F^2} \sum_{i=\floor{na}+1}^{\floor{nb}-1}  \sum_{j_1,j_2 =  \floor{na}+1}^i E[X_{i+1}^T X_{j_1}X_{j_2}^TX_{i+1} e_{i+1}^2e_{j_1}e_{j_2} |\mathcal{F}_i]
\\
& = \frac{1}{n^2||\Sigma||_F^2} \sum_{i=\floor{na}+1}^{\floor{nb}-1}  \sum_{j_1,j_2 =  \floor{na}+1}^i  Z_{j_2}^T H\Big(\frac{j_2}{n}\Big)H\Big(\frac{i+1}{n}\Big)\Sigma H\Big(\frac{i+1}{n}\Big) H\Big(\frac{j_1}{n}\Big)Z_{j_1} e_{j_1}e_{j_2}
\\
& = \frac{1}{n^2||\Sigma||_F^2} \sum_{i=\floor{na}+1}^{\floor{nb}-1}  \sum_{j =  \floor{na}+1}^i  Z_{j}^T H\Big(\frac{j}{n}\Big)H\Big(\frac{i+1}{n}\Big)\Sigma H\Big(\frac{i+1}{n}\Big)H\Big(\frac{j}{n}\Big)Z_{j} e_{j}^2
\\
& \quad + \frac{1}{n^2||\Sigma||_F^2} \sum_{i=\floor{na}+1}^{\floor{nb}-1} 
\sum_{\stackrel{j_1,j_2 =  \floor{na}+1}{j_1 \neq j_2}}^i  Z_{j_2}^TH\Big(\frac{j_2}{n}\Big)H\Big(\frac{i+1}{n}\Big)\Sigma H\Big(\frac{i+1}{n}\Big)H\Big(\frac{j_1}{n}\Big)Z_{j_1} e_{j_1}e_{j_2}
\\ 
& = M_2^{(1)} + M_2^{(2)}.
\end{align*}
For $ M_2^{(1)}$,
\begin{align*}
E[M_2^{(1)}] 
& = \frac{1}{n^2||\Sigma||_F^2} \sum_{i=\floor{na}+1}^{\floor{nb}-1}  \sum_{j =  \floor{na}+1}^i E\Big[Z_{j}^T H\Big(\frac{j}{n}\Big)H\Big(\frac{i+1}{n}\Big)\Sigma H\Big(\frac{i+1}{n}\Big)H\Big(\frac{j}{n}\Big)Z_{j}\Big]
\\
& = \frac{1}{n^2\|\Sigma\|_F^2} \sum_{i=\floor{na}+1}^{\floor{nb}-1}  \sum_{j =  \floor{na}+1}^i \tr\Big(E\Big[ H\Big(\frac{j}{n}\Big)H\Big(\frac{i+1}{n}\Big)\Sigma H\Big(\frac{i+1}{n}\Big)H\Big(\frac{j}{n}\Big)Z_{j}Z_{j}^T\Big]\Big)
\\
& = \frac{1}{n^2\|\Sigma\|_F^2} \sum_{i=\floor{na}+1}^{\floor{nb}-1}  \sum_{j =  \floor{na}+1}^i \tr\Big( H\Big(\frac{j}{n}\Big)H\Big(\frac{i+1}{n}\Big)\Sigma H\Big(\frac{i+1}{n}\Big)H\Big(\frac{j}{n}\Big)\Sigma\Big)
\\
& = \frac{1}{n^2||\Sigma||_F^2} \sum_{i=\floor{na}+1}^{\floor{nb}-1}  \sum_{j =  \floor{na}+1}^i \tr\Big( H^2\Big(\frac{j}{n}\Big)H^2\Big(\frac{i+1}{n}\Big)\Sigma^2\Big)
\\ 
& \to V(a,b).
\end{align*}
Here the last equality follows since for symmetric matrices $A,B,C$ we have 
\[
tr(ABC) = tr((ABC)^T) = tr(CBA) = tr(ACB)
\] 
and since $H(\cdot)$ are diagonal matrices. Thus letting $A = H(j/n)H((i+1)/n)\Sigma, B = H(j/n)H((i+1)/n), C = \Sigma$ the claim follows after some simple computations. Next observe
\begin{align*}
& E[(M_2^{(1)})^2]\\
& =  \frac{1}{n^4\|\Sigma\|_F^4}\sum_{i_1,i_2=\floor{na}+1}^{\floor{nb}-1}  \sum_{j =  \floor{na}+1 }^{\min(i_1,i_2)} \\
& \hspace{2cm} E\Big[Z_{j}^T H\Big(\frac{j}{n}\Big)H\Big(\frac{i_1+1}{n}\Big) \Sigma H\Big(\frac{i_1+1}{n}\Big)H\Big(\frac{j}{n}\Big) Z_{j}Z_{j}^T H\Big(\frac{j}{n}\Big)H\Big(\frac{i_2+1}{n}\Big)\Sigma H\Big(\frac{i_2+1}{n}\Big)H\Big(\frac{j}{n}\Big) Z_{j}\Big]E[e_j^4]\\
& \quad + \frac{1}{n^4\|\Sigma\|_F^4}\sum_{i_1,i_2=\floor{na}+1}^{\floor{nb}-1}  \sum_{j_1 =  \floor{na}+1 }^{i_1} \sum_{\stackrel{j_2 =  \floor{na}+1}{j_1\neq j_2}  }^{i_2} \bigg( E\Big[Z_{j_1}^T H\Big(\frac{j_1}{n}\Big)H\Big(\frac{i_1+1}{n}\Big)\Sigma H\Big(\frac{i_1+1}{n})H\Big(\frac{j_1}{n}\Big) Z_{j_1}\Big] 
\\
& \hspace{6cm} \times E\Big[Z_{j_2}^T H\Big(\frac{j_2}{n}\Big)H\Big(\frac{i_2+1}{n}\Big)\Sigma H\Big(\frac{i_2+1}{n}\Big)H\Big(\frac{j_2}{n}\Big) Z_{j_2}\Big]E[e_{j_1}^2]E[e_{j_2}^2]\bigg).
\end{align*}
For the first part, let 
\[
\tilde \Sigma^{i,j} := H\Big(\frac{j}{n}\Big)H\Big(\frac{i_1+1}{n}\Big) \Sigma H\Big(\frac{i_1+1}{n}\Big)H\Big(\frac{j}{n}\Big)
\]
denote a matrix with entries $\tilde \sigma_{k,l}^{i,j}$. Then
\begin{align*}
& \Big|E\Big[Z_{j}^T H\Big(\frac{j}{n}\Big)H\Big(\frac{i_1+1}{n}\Big) \Sigma H\Big(\frac{i_1+1}{n}\Big)H\Big(\frac{j}{n}\Big) Z_{j}Z_{j}^T H\Big(\frac{j}{n}\Big)H\Big(\frac{i_2+1}{n}\Big)\Sigma H\Big(\frac{i_2+1}{n}\Big)H\Big(\frac{j}{n}\Big) Z_{j}\Big]\Big|
\\
&= \Big|E\Big[\Big(Z_{j}^T H\Big(\frac{j}{n}\Big)H\Big(\frac{i_1+1}{n}\Big) \Sigma H\Big(\frac{i_1+1}{n}\Big)H\Big(\frac{j}{n}\Big) Z_{j}\Big)^2\Big]\Big|
\\
&\le \sum_{k,\ell,k',\ell'=1}^p \Big|\tilde \sigma_{k,l}^{i,j} \tilde \sigma_{k',l'}^{i,j} E\Big[Z_kZ_\ell Z_{k'} Z_{\ell'}  \Big]\Big|
\\
& \le B^8  \sum_{k,\ell,k',\ell'=1}^p \sigma_{k,l}\sigma_{k',l'}\Big|cum(Z_{1,k}Z_{1,l}Z_{1,k'}Z_{1,l'}) + \sigma_{k,l}\sigma_{k',l'} + \sigma_{k,l'}\sigma_{k',l} + \sigma_{k,l'}\sigma_{k',l} \Big|
\\
& \lesssim \|\Sigma\|_F^4, 
\end{align*}
where the last inequality follows from Assumption~\ref{as:2} by repeated application of the Cauchy-Schwarz inequality. For instance
\begin{align*}
\sum_{k,\ell,k',\ell'=1}^p \sigma_{k,l}\sigma_{k',l'}\sigma_{k,l'}\sigma_{k',l} 
\le \Big(\sum_{k,\ell,k',\ell'=1}^p \sigma_{k,l}^2\sigma_{k',l'}^2\Big)^{1/2}\Big(\sum_{k,\ell,k',\ell'=1}^p \sigma_{k,l'}^2\sigma_{k',l}^2 \Big)^{1/2} = \Big(\sum_{k,l=1}^p \sigma_{k,l}^2 \Big)^2 \le C^2 \|\Sigma \|_F^4
\end{align*}
since $\sigma_{k,l} = cum(Z_{1,k},Z_{1,l})$ and similarly for the other terms. The second sum in the representation of $E[(M_2^{(1)})^2]$ can be rewritten as 
\begin{align*}
&\frac{1}{n^4\|\Sigma\|_F^4}\sum_{i_1,i_2=\floor{na}+1}^{\floor{nb}-1}  \sum_{j_1 =  \floor{na}+1 }^{i_1} \sum_{\stackrel{j_2 =  \floor{na}+1}{j_1\neq j_2}  }^{i_2}   \tr\Big(H^2\Big(\frac{j_1}{n}\Big)H^2\Big(\frac{i_1+1}{n}\Big)\Sigma^2\Big)
\tr\Big(H^2\Big(\frac{j_2}{n}\Big)H^2\Big(\frac{i_2+1}{n}\Big)\Sigma^2 \Big)
\\
= & \frac{1}{n^4\|\Sigma\|_F^4}\Big(\sum_{i=\floor{na}+1}^{\floor{nb}-1}  \sum_{j =  \floor{na}+1 }^{i}   \tr\Big(H^2\Big(\frac{j}{n}\Big)H^2\Big(\frac{i+1}{n}\Big)\Sigma^2\Big) \Big)^2
\\
& - \frac{1}{n^4\|\Sigma\|_F^4}\sum_{i_1,i_2=\floor{na}+1}^{\floor{nb}-1}  \sum_{j =  \floor{na}+1 }^{\min(i_1,i_2)}  \tr\Big(H^2\Big(\frac{j}{n}\Big)H^2\Big(\frac{i_1+1}{n}\Big)\Sigma^2\Big)
\tr\Big(H^2\Big(\frac{j}{n}\Big)H^2\Big(\frac{i_2+1}{n}\Big)\Sigma^2 \Big)
\\
= & \frac{1}{n^4\|\Sigma\|_F^4}\Big(\sum_{i=\floor{na}+1}^{\floor{nb}-1}  \sum_{j =  \floor{na}+1 }^{i}   \tr\Big(H^2\Big(\frac{j}{n}\Big)H^2\Big(\frac{i+1}{n}\Big)\Sigma^2\Big) \Big)^2 + O(n^{-1})
\\
\to& V(a,b)^2
\end{align*}
where we used the fact that 
\[
\Big|\tr\Big(H^2\Big(\frac{j}{n}\Big)H^2\Big(\frac{i_2+1}{n}\Big)\Sigma^2 \Big) \Big| \leq B^4 tr(\Sigma^2) = B^4 \|\Sigma\|_F^2.
\]
Therefore, $M_2^{(1)} \xrightarrow{p} V(a,b)$. As for $M_2^{(2)}$, it is easy to check that $E[M_2^{(2)}] = 0$. 
\begin{align*}
& E[(M_2^{(2)})^2]\\
& = \frac{1}{n^4\|\Sigma\|_F^4}\sum_{i,i'=\floor{na}+1}^{\floor{nb}-1}  \sum_{\stackrel{j_1,j_2 =  \floor{na}+1}{j_1\neq j_2} }^i\sum_{\stackrel{j_3,j_4 =  \floor{na}+1}{j_3\neq j_4} }^{i'} \\
& \hspace{2cm}\bigg( E\Big[X_{j_1}^T H\Big(\frac{j_1}{n}\Big)H\Big(\frac{i+1}{n}\Big)\Sigma H\Big(\frac{i+1}{n}\Big)H\Big(\frac{j_2}{n}\Big) X_{j_2}X_{j_3}^T H\Big(\frac{j_3}{n}\Big)H\Big(\frac{i'+1}{n}\Big)\Sigma H\Big(\frac{i'+1}{n}\Big)H\Big(\frac{j_4}{n}\Big) X_{j_4}\Big]
\\
& \hspace{2cm}\times E[e_{j_1}e_{j_2}e_{j_3}e_{j_4}]\bigg)\\
& = \frac{4}{n^4\|\Sigma\|_F^4}\sum_{i,i'=\floor{na}+1}^{\floor{nb}-1}  \sum_{\stackrel{j_1,j_2 =  \floor{na}+1}{j_1< j_2} }^i\sum_{\stackrel{j_3,j_4 =  \floor{na}+1}{j_3< j_4} }^{i'} 
\\
& \hspace{2cm}\bigg( E\Big[X_{j_1}^T H\Big(\frac{j_1}{n}\Big)H\Big(\frac{i+1}{n}\Big)\Sigma H\Big(\frac{i+1}{n}\Big)H\Big(\frac{j_2}{n}\Big) X_{j_2}X_{j_3}^T H\Big(\frac{j_3}{n}\Big)H\Big(\frac{i'+1}{n}\Big)\Sigma H\Big(\frac{i'+1}{n}\Big)H\Big(\frac{j_4}{n}\Big) X_{j_4}\Big]
\\
& \hspace{2cm}\times E[e_{j_1}e_{j_2}e_{j_3}e_{j_4}]\bigg).
\end{align*}
Only when $j_1 = j_3$, $j_2 = j_4$, the expectation is nonzero. Therefore, 
\begin{align*}
& E[(M_2^{(2)})^2]
\\
& = \frac{4}{n^4\|\Sigma\|_F^4}\sum_{i,i'=\floor{na}+1}^{\floor{nb}-1}  \sum_{j_1, j_2 =\floor{na}+1, , j_1 < j_2}^{\min(i,i')}
\\
& \hspace{3cm}E\Big[X_{j_1}^T H\Big(\frac{j_1}{n}\Big)H\Big(\frac{i+1}{n}\Big)\Sigma H\Big(\frac{i+1}{n}\Big)H\Big(\frac{j_2}{n}\Big) X_{j_2}X_{j_2}^T H\Big(\frac{j_2}{n}\Big)H\Big(\frac{i'+1}{n}\Big)\Sigma H\Big(\frac{i'+1}{n}\Big)H\Big(\frac{j_1}{n}\Big) X_{j_1}\Big] 
\\
& \le \frac{4}{n^4\|\Sigma\|_F^4}\sum_{i,i'=\floor{na}+1}^{\floor{nb}-1}  \sum_{j_1, j_2 =\floor{na}+1, , j_1 < j_2}^{\min(i,i')} B^8 tr(\Sigma^4) 
\\
& = \frac{4}{\|\Sigma\|_F^4} tr(\Sigma^4) O(1)
\\
& = O\left(\frac{ tr(\Sigma^4)}{\|\Sigma\|^4_F}\right) \to 0.
\end{align*}
Here the inequality follows since by repeated application of the identity
\[
tr(ABC) = tr((ABC)^T) = tr(CBA) = tr(ACB) = tr(BAC)
\]
valid for symmetric matrices $A,B,C$ as well as the cyclic permutation property of the trace operator we have
\begin{align*}
&tr\Big(\Sigma H\Big(\frac{j_1}{n}\Big)H\Big(\frac{i+1}{n}\Big)\Sigma H\Big(\frac{i+1}{n}\Big)H\Big(\frac{j_2}{n}\Big) \Sigma H\Big(\frac{j_2}{n}\Big)H\Big(\frac{i'+1}{n}\Big)\Sigma H\Big(\frac{i'+1}{n}\Big)H\Big(\frac{j_1}{n}\Big) \Big)
\\
=& tr\Big(\Sigma^2 H\Big(\frac{j_1}{n}\Big)H\Big(\frac{i+1}{n}\Big)\Sigma H\Big(\frac{i+1}{n}\Big)H\Big(\frac{j_2}{n}\Big) \Sigma H\Big(\frac{j_2}{n}\Big)H\Big(\frac{i'+1}{n}\Big)\Sigma H\Big(\frac{i'+1}{n}\Big)H\Big(\frac{j_1}{n}\Big) \Big) 
\\
=& ... = tr\Big(\Sigma^4 H\Big(\frac{j_1}{n}\Big)H\Big(\frac{i+1}{n}\Big) H\Big(\frac{i+1}{n}\Big)H\Big(\frac{j_2}{n}\Big)  H\Big(\frac{j_2}{n}\Big)H\Big(\frac{i'+1}{n}\Big) H\Big(\frac{i'+1}{n}\Big)H\Big(\frac{j_1}{n}\Big) \Big).
\end{align*}
Now the largest entry of the diagonal matrix 
\[
H\Big(\frac{j_1}{n}\Big)H\Big(\frac{i+1}{n}\Big) H\Big(\frac{i+1}{n}\Big)H\Big(\frac{j_2}{n}\Big)  H\Big(\frac{j_2}{n}\Big)H\Big(\frac{i'+1}{n}\Big) H\Big(\frac{i'+1}{n}\Big)H\Big(\frac{j_1}{n}\Big)
\]
is bounded by $B^8$ and $\Sigma^4$ has positive diagonal entries (since it can be seen as covariance matrix of $(\Sigma^{1/2})^3 Z_1$) so that
\[
tr\Big(\Sigma^4 H\Big(\frac{j_1}{n}\Big)H\Big(\frac{i+1}{n}\Big) H\Big(\frac{i+1}{n}\Big)H\Big(\frac{j_2}{n}\Big)  H\Big(\frac{j_2}{n}\Big)H\Big(\frac{i'+1}{n}\Big) H\Big(\frac{i'+1}{n}\Big)H\Big(\frac{j_1}{n}\Big) \Big) \leq B^8 tr(\Sigma^4). 
\]
This shows that $M_2^{(2)} \xrightarrow{p} 0$. Together with previous result, we have shown that 
\[
M_2 \xrightarrow{p} V(a,b). 
\]
Similarly, 
\[
M_3 \xrightarrow{p}V(a,b), \quad M_1 \xrightarrow{p}V(a,b). 
\]
This completes the proof.
\end{proof}


\subsection{Proof of Theorem~\ref{thm:3}: Power of the test}

The following equivalent representation for the quantity $G_n(k)$ will be useful: 
\begin{align}
G_n(k) &= \frac{1}{k(k-1)(n-k)(n-k-1)}\sum_{1\leq j_1,j_3\leq k, j_1\neq j_3} \sum_{k+1\leq j_2,j_4\leq n, j_2\neq j_4} (X_{j_1} - X_{j_2})^T(X_{j_3} - X_{j_4}) \notag
\\
&= \frac{1}{k(k-1)(n-k)(n-k-1)} D_n(k) \label{eq:GnDn}
\end{align}
This expression can be obtained by elementary calculations after multiplying out the products in the expression above.

Further, recall that we assumed $k^* = \floor{nc}$ for some constant $c \in (0,1)$. Define a new sequence of random vectors $Y_i$,
\[Y_ i = H(i/n)Z_i = \begin{cases} X_i & i = 1,\ldots,k^*\\
X_i - \Delta & i = k^*+1,\ldots, n
\end{cases}. \]
This sequence does not have a change point. Without loss of generosity, assume $Y_i$'s are centered.
The remaining proof consists of a detailed analysis of the original test statistic and the bootstrap statistic under different types of alternatives.

For the bootstrap statistic, we will prove that under the null and any alternative $S_n^{*X}$ satisfies 
\begin{equation} \label{eq:SnStaralt}
\frac{S_n^{*X}(a,b)}{n\|\Sigma\|_F} = \frac{S_n^{*Y}(a,b)}{n\|\Sigma\|_F} + O_p\Big(\frac{(\Delta^T \Sigma \Delta)^{1/2}}{\|\Sigma\|_F}\Big) + O_p\Big(\frac{\|\Delta\|_2^2}{\|\Sigma\|_F}\Big)
\end{equation}
where the remainder terms are uniform in $a,b \in [0,1]$ and $S_n^{*Y}$ is defined in exactly the same way as $S_n^{*X}$ but with $Y_i$ in place of $X_i$. We will further show that
\begin{equation}\label{eq:Tnalt}
T_n \left\{
\begin{array}{ll}
\xrightarrow{d} T &, n\|\Delta\|_2^2/\|\Sigma\|_F \to 0
\\
\xrightarrow{d} \sup_{r \in [0,1]} \{G(r) + \Lambda(r)\} &, n\|\Delta\|_2^2/\|\Sigma\|_F \to \beta \in (0,\infty)
\\
\geq \frac{k^*(k^*-1)(n-k^*)(n-k^*-1)}{n^4} \frac{n\|\Delta\|_2^2}{\|\Sigma\|_F} + o_P\Big(n\frac{\|\Delta\|_2^2}{\|\Sigma\|_F}\Big) &, n\|\Delta\|_2^2/\|\Sigma\|_F \to \infty
\end{array}
\right.
\end{equation}
where
\[
\Lambda(r) = \begin{cases} (1-c)^2r^2\beta & r \le c\\
c^2 (1-r)^2 \beta & r > c
\end{cases}
\]
for $c = \lim_{n\to \infty} k^*/n$. The argument in the case $n\|\Delta\|_2^2/\|\Sigma\|_F \to \beta \in (0,\infty)$ is complete. The remaining two cases are discussed below.

\textbf{The case $n\|\Delta\|_2^2/\|\Sigma\|_F \to 0$} In this case~\eqref{eq:SnStaralt} implies
\[
\frac{S_n^{*X}(a,b)}{n\|\Sigma\|_F} = \frac{S_n^{*Y}(a,b)}{n\|\Sigma\|_F} + o_P(1).
\]
Since the sequence $Y_i$ contains no change-points and satisfies all assumptions of Theorem~\ref{thm:2}, the proof follows from exactly the same arguments as the proof of the latter result.

\medskip

\textbf{The case $n\|\Delta\|_2^2/\|\Sigma\|_F \to \infty$}: from expression~\eqref{eq:SnStaralt} we find that in this case $T_n^* = o_P(n\|\Delta\|_2^2/\|\Sigma\|_F)$ and hence $c_{1,\alpha}^{(M)} = o_P(n\|\Delta\|_2^2/\|\Sigma\|_F)$. Since also 
\[
\frac{k^*(k^*-1)(n-k^*)(n-k^*-1)}{n^4} \to c^2(1-c)^2 
\]
we obtain
\[
P\Big(T_n > c_{1,\alpha}^{(M)}\Big) \geq P\Big( c^2(1-c)^2 n\|\Delta\|_2^2/\|\Sigma\|_F > o_P(n\|\Delta\|_2^2/\|\Sigma\|_F) \Big) \to 1.
\]
This completes the proof of Theorem~\ref{thm:3}. \hfill $\Box$


\subsubsection{Proof of~\eqref{eq:Tnalt}: Behaviour of $G_n(k)$ under the alternative}

Simple computations show that in the case $k^* > k$, the statistic $D_n(k)$ admits the following decomposition
\begin{align*}
D_n(k) = & D_n^Y(k) + k(k-1)(n-k^*)(n-k^*-1) \|\Delta\|_2^2 - 2(k-1)(n-k^*)(n-k-2) \sum_{j=1}^k Y_j^T \Delta
\\
& - 4(k-1)(k-2)(n-k^*)\sum_{j=k+1}^{k^*} Y_j^T \Delta
\end{align*} 
where $D_n^Y$ is defined similarly as $D_n$ but with $Y_i$ replacing $X_i$. By Kolmogorov's maximal inequality we have
\[
\sup_{1 \le l\le k \le  n}\Big|\sum_{j = l}^k X_j^T \Delta \Big| \le 2 \sup_{1\le k \le n}\Big|\sum_{j = 1}^k X_j^T \Delta \Big|  = O_p(n^{1/2} (\Delta^T \Sigma \Delta)^{1/2}) 
= o_p(n^{1/2} \|\Delta\|_2 \|\Sigma\|_F^{1/2})
\]
where we used the bound
\[
\max_j Var\Big(\Delta^T Y_j\Big) \le B^2 \Delta^T \Sigma \Delta
\]
and the fact that under Assumption 1 we have $\|\Sigma\|_2 = o(\|\Sigma\|_F)$, see Remark 3.2 in \cite{wang2019inference}.

Combining this with the representation in~\eqref{eq:GnDn} we find that, uniformly in $1 \leq k \leq k^*$,
\begin{align*}
\tilde G_n(k) = \tilde G_n^Y(k) + \frac{k(k-1)(n-k^*)(n-k^*-1) \|\Delta\|_2^2}{n^3} + o_p(n^{1/2} \|\Delta\|_2 \|\Sigma\|_F^{1/2}). 
\end{align*}
Similar arguments show that, uniformly in $1 \leq k^* \leq k \leq n$,
\begin{align*}
\tilde G_n(k) = \tilde G_n^Y(k) + \frac{k^*(k^*-1)(n-k)(n-k-1) \|\Delta\|_2^2}{n^3} + o_p(n^{1/2} \|\Delta\|_2 \|\Sigma\|_F^{1/2}). 
\end{align*}

 Finally, elementary computations show that for $k/n \to r$ 
\[
\frac{k(k-1)(n-k^*)(n-k^*-1)}{n^4} \to r^2(1-c)^2
\]
and 
\[
\frac{k^*(k^*-1)(n-k)(n-k-1)}{n^4} \to c^2(1-r)^2
\]
We now discuss the consequence of this result for three types of alternatives.

\medskip

\noindent
\textbf{case 1: $n\|\Delta\|_2^2/\|\Sigma\|_F \to 0$} 

In this case we have
\[
\frac{\tilde G_n(k)}{\|\Sigma\|_F} = \frac{\tilde G_n^Y(k)}{\|\Sigma\|_F} + o_P(1)
\] 
uniformly in $k$. Hence $T_n \Dkonv T$.

\medskip
\noindent
\textbf{case 2: $n\|\Delta\|_2^2/\|\Sigma\|_F \to \beta > 0$}

In this case we obtain
\[
\Big(\frac{\tilde G_n(\floor{nr})}{\|\Sigma\|_F}\Big)_{r \in [0,1]} \weak \Big( G(r) + \Lambda(r)\Big)_{r \in [0,1]}
\]
where 
\[
\Lambda(r) = \begin{cases} (1-c)^2r^2\beta & r \le c\\
c^2 (1-r)^2 \beta & r > c
\end{cases}.
\]
Hence by the continuous mapping theorem
\[
T_n \Dkonv \sup_{r \in [0,1]} G(r) + \Lambda(r).
\]
\medskip
\noindent
\textbf{case 3: $n\|\Delta\|_2^2/\|\Sigma\|_F \to \infty$} 

In this case note that
\[
T_n \ge \frac{\Tilde{G}_n(k^*) }{\|\Sigma\|_F} = O_P(1) + \frac{k^*(k^*-1)(n-k^*)(n-k^*-1)}{n^4} \frac{n\|\Delta\|_2^2}{\|\Sigma\|_F} + o_P\Big(\frac{n\|\Delta\|_2^2}{\|\Sigma\|_F}\Big). 
\]
This completes the proof of~\eqref{eq:Tnalt} \hfill $\Box$

\subsubsection{Proof of~\eqref{eq:SnStaralt}: $S_n^*(a,b)$ under the alternatives}
For the bootstrap partial sum process, we observe the following decomposition:
\begin{align*}
S_n^{*X}(a,b) &= \sum_{i = \floor{na}+1}^{\floor{nb}-1} \sum_{j= \floor{na}+1}^i (X_{i+1} - \bar X)^T (X_j - \bar X) e_{i+1}e_j\\
& = \sum_{i = \floor{na}+1}^{\floor{nb}-1} \sum_{j= \floor{na}+1}^i \left(\Delta[\bm I( k^*+1\le i+1 \le n) - \frac{n-k^*}{n}] + Y_{i+1} - \bar Y\right)^T \\
& \hspace{5cm}\left( \Delta[\bm I( k^*+1\le j \le n) - \frac{n-k^*}{n}] + Y_j - \bar Y\right) e_{i+1}e_j\\
& = \sum_{i = \floor{na}+1}^{\floor{nb}-1} \sum_{j= \floor{na}+1}^i (Y_{i+1} - \bar Y)^T (Y_j - \bar Y) e_{i+1}e_j \\
& \qquad + \sum_{i = \floor{na}+1}^{\floor{nb}-1} \sum_{j= \floor{na}+1}^i \Delta^T(Y_{j} - \bar Y) [\bm I( k^*+1\le i+1 \le n) - \frac{n-k^*}{n}] e_{i+1}e_j
\\ 
& \qquad + \sum_{i = \floor{na}+1}^{\floor{nb}-1} \sum_{j= \floor{na}+1}^i \Delta^T(Y_{i+1} - \bar Y) [\bm I( k^*+1\le j \le n) - \frac{n-k^*}{n}] e_{i+1}e_j
\\ 
& \qquad + \sum_{i = \floor{na}+1}^{\floor{nb}-1} \sum_{j= \floor{na}+1}^i \Delta^T\Delta [\bm I( k^*+1\le i+1 \le n) - \frac{n-k^*}{n}][\bm I( k^*+1\le j \le n) - \frac{n-k^*}{n}] e_{i+1}e_j
\\
& = S_n^{*Y}(a,b) + S_{n,2}^{*Y}(a,b) + S_{n,3}^{*Y}(a,b) + S_{n,4}^{*Y}(a,b)
\end{align*}
The first term corresponds to the case with no changepoint and has the same limiting behaviour as under the null. We now study he behaviour of the remainder terms. 

\bigskip

\textbf{The case $c < a < b$.} The remainder terms take the form 
\[
\sum_{i = \floor{na}+1}^{\floor{nb}-1} \sum_{j= \floor{na}+1}^i \frac{k^*}{n}\Delta^T(Y_{j} - \bar Y) e_{i+1}e_j 
+ \sum_{i = \floor{na}+1}^{\floor{nb}-1} \sum_{j= \floor{na}+1}^i \frac{k^*}{n}\Delta^T(Y_{i+1} - \bar Y) e_{i+1}e_j
+ \sum_{i = \floor{na}+1}^{\floor{nb}-1} \sum_{j= \floor{na}+1}^i \Big(\frac{k^*}{n}\Big)^2\Delta^T\Delta e_{i+1}e_j
\]
For the last term, observe that this has the same form as $S_n(a,b)/n$ where $X_i$ are replaced by $e_i$. The corresponding covariance matrix is $\Sigma =1$ and hence by weak convergence of $S_n/n\|\Sigma\|_F$ under general conditions which are satisfied in this special case we have
\begin{equation}\label{eq:supsumeiej}
\sup_{a,b} \frac{1}{n} \sum_{\stackrel{i,j = \floor{na}+1}{i \neq j}}^{\floor{nb}-1} e_{i}e_j = O_p(1).
\end{equation}
Therefore, the last term is of order $O_p(n\|\Delta\|_2^2)$.

The terms $S_{n,2}^{*Y}(a,b)$ and $S_{n,3}^{*Y}(a,b)$ will be handled together. Note that
\begin{align*}
S_{n,2}^{*Y}(a,b) + S_{n,3}^{*Y}(a,b) 
& = \sum_{i = \floor{na}+1}^{\floor{nb}-1} \sum_{j= \floor{na}+1}^i \frac{k^*}{n}\Delta^T(Y_{j} + Y_{i+1} - 2\bar Y) e_{i+1}e_j
\\
& = \frac{1}{2}\sum_{i,j = \floor{na}+1, i \neq j}^{\floor{nb}-1} \frac{k^*}{n}\Delta^T(Y_{j} + Y_{i} - 2\bar Y) e_{i}e_j
\\
& = \frac{k^*}{2n}\Big( \sum_{i,j = \floor{na}+1}^{\floor{nb}-1} \Delta^T(Y_{j} + Y_{i}) e_{i}e_j - \sum_{i = \floor{na}+1}^{\floor{nb}-1} 2 \Delta^T Y_i e_i^2 \Big)- \Delta^T\bar Y \sum_{i,j = \floor{na}+1, i \neq j}^{\floor{nb}-1} \frac{k^*}{n} e_{i}e_j.
\end{align*}
For the first term in the bracket observe that
\begin{align*}
\sum_{i,j = \floor{na}+1}^{\floor{nb}-1} \Delta^T Y_{j} e_{i}e_j = \left\{\sum_{i =  \floor{na}+1}^{\floor{nb}-1} \Delta^T Y_{i} e_{i}\right\} \left\{\sum_{j =  \floor{na}+1}^{\floor{nb}-1} e_j \right\}. 
\end{align*}
Since $E[Y_i] = 0$ and $Y_i$ are independent of $e_i$ we obtain by Kolmogorov's maximal inequality, 
\begin{align}
\sup_{0 \leq a,b \leq 1} \left|\sum_{i =  \floor{na}+1}^{\floor{nb}-1}  \Delta^T Y_{i} e_{i}\right|  \label{eq:hpp1}
&= O_p(n^{1/2} \max_jVar(\Delta^T Y_{j} e_{j})^{1/2} ) 
= O_p(n^{1/2}(\Delta^T \Sigma \Delta)^{1/2}),
\\
\sup_{0 \leq a,b \leq 1} \left|\sum_{i =  \floor{na}+1}^{\floor{nb}-1}  \Delta^T Y_{i} e_{i}^2\right| 
&= O_p(n^{1/2} \max_jVar(\Delta^T Y_{j} e_{j}^2)^{1/2} ) = O_p(n^{1/2}(\Delta^T \Sigma \Delta)^{1/2}), \label{eq:hpp2}
\\
\sup_{0 \leq a,b \leq 1} \left|\sum_{j =  \floor{na}+1}^{\floor{nb}-1} e_j \right| 
&= O_p(n^{1/2}). \label{eq:hpp3}
\end{align}
For the last term we have by~\eqref{eq:supsumeiej} and an elementary calculation using independence of the $Y_i$
\begin{align}
\sup_{0 \leq a,b \leq 1} \left|\Delta^T\bar Y \sum_{i,j = \floor{na}+1, i \neq j}^{\floor{nb}-1} \frac{k^*}{n} e_{i}e_j \right| \notag
&\leq  |\Delta^T\bar Y| \cdot\sup_{a,b} \left|\sum_{i = \floor{na}+1}^{\floor{nb}-1} \sum_{j= \floor{na}+1}^ie_{i+1}e_j\right| 
\\
&=  O_p(Var(\Delta^T\bar Y )^{1/2})O_p(n) \notag
\\
& = O_p(n^{1/2}(\Delta^T \Sigma \Delta)^{1/2}). \label{eq:hpp4}
\end{align}
In summary, we have proved that in the case $c < a < b$
\begin{equation}
\frac{S_n^{*X}(a,b)}{n\|\Sigma\|_F} = \frac{S_n^{*Y}(a,b)}{n\|\Sigma\|_F} + O_p\Big(\frac{(\Delta^T \Sigma \Delta)^{1/2}}{\|\Sigma\|_F}\Big) + O_p\Big(\frac{\|\Delta\|_2^2}{\|\Sigma\|_F}\Big).
\end{equation}

\textbf{The case $ a < b < c$.} The remainder terms take the form 
\begin{align*}
- \sum_{i = \floor{na}+1}^{\floor{nb}-1} \sum_{j= \floor{na}+1}^i \frac{n-k^*}{n}\Delta^T(Y_{j} - \bar Y) e_{i+1}e_j 
& - \sum_{i = \floor{na}+1}^{\floor{nb}-1} \sum_{j= \floor{na}+1}^i \frac{n-k^*}{n}\Delta^T(Y_{i+1} - \bar Y) e_{i+1}e_j
\\
& + \sum_{i = \floor{na}+1}^{\floor{nb}-1} \sum_{j= \floor{na}+1}^i \Big(\frac{n-k^*}{n}\Big)^2\Delta^T\Delta e_{i+1}e_j.
\end{align*}
This can be handled similarly to the case $c < a <b$.
\\ 

\bigskip
\textbf{The case $ a < c < b $.} 
Compared to the case $a < b < c$ we have the additional terms
\begin{align*}
    \sum_{i = \floor{na}+1}^{\floor{nb}-1} \sum_{j= \floor{na}+1}^i \Delta^T(Y_{j} - \bar Y) w_{i+1} e_{i+1}e_j  +    \sum_{i = \floor{na}+1}^{\floor{nb}-1} \sum_{j= \floor{na}+1}^i \Delta^T(Y_{i+1} - \bar Y) w_{j} e_{i+1}e_j
    \\
    +     \sum_{i = \floor{na}+1}^{\floor{nb}-1} \sum_{j= \floor{na}+1}^i \Delta^T\Delta w_{i+1}w_j e_{i+1}e_j
\end{align*}
where $w_t = \bm I \{k^* + 1\le t \le n\}$. .
For the last term, note that
\[
w_t = H(t/n) + r_n(t)
\]
where $H(x) = \bm I \{c \le x \}$ and $|r_n(t)| \leq \bm I \{ |t - \floor{nc}| \leq 1 \}$. Now a direct computation shows that the pieces involving $r_n(t)$ are negligible while the remaining term takes the form
\[
\sum_{i = \floor{na}+1}^{\floor{nb}-1} \sum_{j= \floor{na}+1}^i \Delta^T\Delta H((i+1)/n)H(j/n) e_{i+1}e_j
\]
This has the same form as $S_n(a,b)/n$ where $\mu_i=0$, $H$ as above and $Z_i$ are replaced by $e_i$. The corresponding covariance matrix is $\Sigma =1$ and hence by weak convergence of $S_n/n\|\Sigma\|_F$ under general conditions which are satisfied in this special case we have
\begin{equation}\label{eq:supsumeiej}
\sup_{a,b} \frac{1}{n} \sum_{\stackrel{i,j = \floor{na}+1}{i \neq j}}^{\floor{nb}-1} H((i+1)/n)H(j/n) e_{i+1}e_j = O_p(1).
\end{equation}
Therefore, the last term is of order $O_p(n\|\Delta\|_2^2)$.

Next we bound the first two terms. We have
\begin{align*}
& \sum_{i = \floor{na}+1}^{\floor{nb}-1} \sum_{j= \floor{na}+1}^i \Delta^T(Y_{j} - \bar Y) w_{i+1} e_{i+1}e_j  +    \sum_{i = \floor{na}+1}^{\floor{nb}-1} \sum_{j= \floor{na}+1}^i \Delta^T(Y_{i+1} - \bar Y) w_{j} e_{i+1}e_j
\\
= & \sum_{i = \floor{nc}}^{\floor{nb}-1} \sum_{j= \floor{nc}}^i \Delta^T( Y_{j} - \bar Y) e_{i+1}e_j + \sum_{i = \floor{nc}}^{\floor{nb}-1} \sum_{j= \floor{na} + 1}^{\floor{nc}-1}\Delta^T( Y_{j} - \bar Y) e_{i+1}e_j
\\
& + \sum_{i = \floor{nc}}^{\floor{nb}-1} \sum_{j= \floor{nc}}^i \Delta^T(Y_{i+1} - \bar Y)  e_{i+1}e_j - \sum_{i = \floor{nc}}^{\floor{nb}-1}  \Delta^T(Y_{i+1} - \bar Y)  e_{i+1}e_{\floor{nc}}.
\end{align*}
The terms 
\[
\sum_{i = \floor{nc}}^{\floor{nb}-1} \sum_{j= \floor{nc}}^i \Delta^T( Y_{j} - \bar Y) e_{i+1}e_j + \sum_{i = \floor{nc}}^{\floor{nb}-1} \sum_{j= \floor{nc}}^i \Delta^T(Y_{i+1} - \bar Y)  e_{i+1}e_j
\]
have a similar structure as in the case $a < b <c$ and can be treated similarly as there. For the remaining two terms note that
\begin{align*}
&\Big| \sum_{i = \floor{nc}}^{\floor{nb}-1} \sum_{j= \floor{na} + 1}^{\floor{nc}-1}\Delta^T( Y_{j} - \bar Y) e_{i+1}e_j - e_{\floor{nc}} \sum_{i = \floor{nc}}^{\floor{nb}-1}  \Delta^T(Y_{i+1} - \bar Y)  e_{i+1} \Big|
\\
\leq & \Big|\sum_{i = \floor{nc}}^{\floor{nb}-1}  e_{i+1} \Big| \cdot \Big| \sum_{j= \floor{na} + 1}^{\floor{nc}-1} \Delta^T( Y_{j} - \bar Y)e_j\Big| + |e_{\floor{nc}}| \cdot \Big| \sum_{i = \floor{nc}}^{\floor{nb}-1}  \Delta^T(Y_{i+1} - \bar Y)  e_{i+1} \Big|
\\
\leq & 2\sup_{0 \leq a,b \leq 1} \left|\sum_{i =  \floor{na}+1}^{\floor{nb}-1}  \Delta^T Y_{i} e_{i}\right| \cdot\sup_{0 \leq a,b \leq 1} \left|\sum_{j =  \floor{na}+1}^{\floor{nb}-1} e_j \right|  
+ 2|\Delta^T\bar Y|\cdot \left(\sup_{0 \leq a,b \leq 1} \left|\sum_{j =  \floor{na}+1}^{\floor{nb}-1} e_j \right|\right )^2
\\
= & O_p(n^{1/2}(\Delta^T \Sigma \Delta)^{1/2}) O_p(n^{1/2}) + O_p(n^{-1/2}(\Delta^T \Sigma \Delta)^{1/2}) O_p(n) = O_p(n(\Delta^T \Sigma \Delta)^{1/2})
\end{align*}
Where the last line uses the bounds in~\eqref{eq:hpp1}-\eqref{eq:hpp4}. Summarizing, we have proved that in the case $ a <c < b$
\begin{equation}
\frac{S_n^{*X}(a,b)}{n\|\Sigma\|_F} = \frac{S_n^{*Y}(a,b)}{n\|\Sigma\|_F} + O_p\Big(\frac{(\Delta^T \Sigma \Delta)^{1/2}}{\|\Sigma\|_F}\Big) + O_p\Big(\frac{\|\Delta\|_2^2}{\|\Sigma\|_F}\Big).
\end{equation}
This completes the proof. \hfill $\Box$

 \subsection{Proof of Theorem~\ref{thm:4} and Theorem~\ref{th:multCppower}: Theory for multiple change point testing}
 Under the null, the process convergence result of $S^*_n(a,b)$ and continuous mapping theorem, we conclude that
\[\frac{T^*_{n,M}}{\|\Sigma\|_F} \xrightarrow{d}\sup_{0\le r_1 < r_2 \le 1} G(r_1;0,r_2) + \sup_{0\le r_1 < r_2 \le 1} G(r_2;r_1,1), \text{ in probability.}\]
 Similar to the arguments in the proof of Theorem 3, the result stated in Theorem 4 holds. \\
 
Under the alternative, there are $M$ change points at locations $k_1, k_2,\ldots,k_M$, and denote the changes by $\Delta_i = \mu_{i+1} - \mu_{i}$. The partial sum process can be decomposed as follows,
\begin{align*}
S_n^{*X}(a,b) &= \sum_{i = \floor{na}+1}^{\floor{nb}-1} \sum_{j= \floor{na}+1}^i (X_{i+1} - \bar X)^T (X_j - \bar X) e_{i+1}e_j
\\
& = \sum_{i = \floor{na}+1}^{\floor{nb}-1} \sum_{j= \floor{na}+1}^i 
\left( \sum_{r = 0}^{M} 
\Delta_r\bm I( k_r \le i+1 \le k_{r+1}) + Y_{i+1} - \bar Y
\right)^T \\ 
& \quad\quad\quad \times \left( \sum_{r = 0}^{M} 
\Delta_r\bm I( k_r \le j \le k_{r+1})+ Y_{j} - \bar Y 
\right)e_{i+1}e_j \\
& = \sum_{i = \floor{na}+1}^{\floor{nb}-1} \sum_{j= \floor{na}+1}^i  (Y_{i+1} - \bar Y
)^T  (Y_{j} - \bar Y
)e_{i+1}e_j \\
& \quad \quad +  \sum_{r = 0}^M\sum_{i = \floor{na}+1}^{\floor{nb}-1} \sum_{j= \floor{na}+1}^i  Y_{j}^T \Delta_r \bm I( k_r \le i+1 \le k_{r+1})e_{i+1}e_j\\
& \quad \quad+ \sum_{r = 0}^M\sum_{i = \floor{na}+1}^{\floor{nb}-1} \sum_{j= \floor{na}+1}^i  Y_{i+1}^T \Delta_r \bm I( k_r \le j \le k_{r+1})e_{i+1}e_j \\
& \quad \quad+ \sum_{t,r = 0}^M\sum_{i = \floor{na}+1}^{\floor{nb}-1} \sum_{j= \floor{na}+1}^i \Delta_t^T \Delta_r \bm I( k_r \le i+1 \le k_{r+1}) \bm I( k_t \le j \le k_{t+1})e_{i+1}e_j.
\end{align*}
Again, the first term is simply the process under the null. Similar analysis to single change point case shows that the second term and third term is of order $O_p(n \sum_{i = 1}^M(\Delta_i^T \Sigma \Delta_i)^{1/2})$. The last term is of order $O_p(n\sum_{i = 1}^M\|\Delta_i\|_2^2)$. 
Under local or fixed alternative, that is 
\[\frac{n\Delta_i^T\Delta_i}{\|\Sigma\|_F} \to b_i \in [0, \infty), \text{ for all $i = 1, \ldots M$},\]
$\frac{S_n^{*X}(a,b)}{n \|\Sigma\|_F}$ converges to the same process under the null. 
When there is at least one diverging change point, according the the bootstrap statistic is bounded by $O_p(\|\Delta_s\|_2^2)$, where $\Delta_s = \mu_{s+1} - \mu_s$ is the largest change. To show that the proposed test has power converging to one, it suffices to check that the order of the original statistic. To this end, we consider the forward scanning statistic $G_n(k_s;1,k_{s+1})$.

\begin{align*}
    & {k_s(k_s-1)(k_{s+1}- k_s)(k_{s+1}- k_s - 1)}G_n(k_s;1,k_{s+1})  \\
    & = \sum_{j_1 = 1}^{k_{s}} \sum_{j_3 = 1, j_3\ne j_1}^{k_s} \sum_{j_2 = k_s+1}^{k_{s+1}} \sum_{j_4 = k_s+1, j_4\ne j_2}^{k_{s+1}} (X_{j_1} - X_{j_2})^T(X_{j_3} - X_{j_4})\\
    & = \sum_{j_1 = 1}^{k_{s}} \sum_{j_3 = 1, j_3\ne j_1}^{k_s} \sum_{j_2 = k_s+1}^{k_{s+1}} \sum_{j_4 = k_s+1, j_4\ne j_2}^{k_{s+1}} (Y_{j_1} - Y_{j_2} + \sum_{j=0}^{s-1} \bm I\{k_j +1 \leq j_1 \leq k_{j+1} \}\mu_{j+1} - \mu_s )^T\\
    & \qquad \times (Y_{j_3} - Y_{j_4}+ \sum_{j=0}^{s-1} \bm I\{k_j +1 \leq j_3 \leq k_{j+1} \}\mu_{j+1} - \mu_s)\\
    & = \sum_{1 \leq j_1 \neq j_3 \leq k_s} \sum_{k_s+1 \leq j_2 \neq j_4 \leq k_{s+1}} \Big(\mu_{s+1} - \sum_{j=0}^{s-1} \bm I\{k_j +1 \leq j_1 \leq k_{j+1} \}\mu_{j+1} \Big)^T \Big(\mu_{s+1} - \sum_{j=1}^s \bm I\{k_j +1 \leq j_3 \leq k_{j+1} \} \mu_{j+1} \Big)\\
    & \qquad + \sum_{1 \leq j_1 \neq j_3 \leq k_s} \sum_{k_s+1 \leq j_2 \neq j_4 \leq k_{s+1}}(Y_{j_1} - Y_{j_2} )^T(Y_{j_3} - Y_{j_4})\\
    & \qquad + \sum_{1 \leq j_1 \neq j_3 \leq k_s} \sum_{k_s+1 \leq j_2 \neq j_4 \leq k_{s+1}}(Y_{j_1} - Y_{j_2} )^T(\sum_{j=0}^{s-1} \bm I\{k_j +1 \leq j_3 \leq k_{j+1} \}\mu_{j+1} - \mu_s) \\
    & \qquad + \sum_{1 \leq j_1 \neq j_3 \leq k_s} \sum_{k_s+1 \leq j_2 \neq j_4 \leq k_{s+1}}(Y_{j_3} - Y_{j_4})^T(\sum_{j=0}^{s-1} \bm I\{k_j +1 \leq j_1 \leq k_{j+1} \}\mu_{j+1} - \mu_s )
\end{align*}
Similar to the single change point case, the order of the first term dominates. Since we assumed that $k_s$ is the largest change, and there are only finite change points, the order of the this is $n^4\|\Delta_s\|_2^2$. After proper scaling, the order of the original test statistic is $T_{n,M} = O_p(n\|\Delta_s\|_2^2)$. Together with the fact that the bootstrap statistic is of order $\|\Delta_s\|_2^2$, we conclude that the power will converge to 1.

\bibliographystyle{chicago}
\bibliography{sample}

\end{document}